\documentclass[runningheads]{llncs}
\usepackage[T1]{fontenc}
\usepackage{graphicx}
\usepackage{microtype}
\usepackage{subfigure}
\usepackage{booktabs} 
\usepackage{multicol}
\usepackage{multirow}
\usepackage{enumitem}
\usepackage{comment}
\usepackage{tablefootnote}
\usepackage{makecell}
\usepackage[hidelinks]{hyperref}
\usepackage{color}

\spnewtheorem*{theorem*}{Theorem}{\normalshape\bfseries}{\itshape}
\usepackage{amsmath}
\usepackage{bm}
\usepackage{amssymb}
\usepackage{mathtools}
\usepackage{algorithm, algorithmicx, algpseudocode, graphicx}

\usepackage[capitalize,noabbrev]{cleveref}

\usepackage{environ}
\newcommand{\repeattheorem}[1]{%
  \begingroup
  \renewcommand{\thetheorem}{\ref{#1}}%
  \expandafter\expandafter\expandafter\theorem
  \csname reptheorem@#1\endcsname
  \endtheorem
  \endgroup
}

\NewEnviron{reptheorem}[1]{%
  \global\expandafter\xdef\csname reptheorem@#1\endcsname{%
    \unexpanded\expandafter{\BODY}%
  }%
  \expandafter\theorem\BODY\unskip\label{#1}\endtheorem
}

\newcommand{\B}{\mathcal{B}}
\newcommand{\U}{\mathcal{U}}
\newcommand{\N}{\mathcal{N}}
\newcommand{\M}{\mathcal{M}}
\newcommand{\ptas}{\textsc{PTAS}}
\newcommand{\A}{\mathcal{A}}
\renewcommand{\S}{\mathcal{S}}
\newcommand{\F}{\mathcal{F}}
\newcommand{\C}{\mathcal{C}}

\newcommand{\tta}{\theta}
\newcommand{\OPT}{\texttt{OPT}}
\newcommand{\ceil}[1]{\lceil {#1} \rceil}
\newcommand{\floor}[1]{\lfloor {#1} \rfloor}

\DeclareMathOperator*{\argmax}{argmax}
\DeclareMathOperator*{\argmin}{argmin}

\newcommand{\s}{\mathbf{s}}
\newcommand{\spc}{\{0,1\}-speed}
\newcommand{\p}{\mathbf{p}}

\newcommand{\OPTC}{\overline{\OPT}_C}

\newcommand{\mpred}{\hat{m}}
\newcommand{\mtrue}{m_0}
\newcommand{\mupp}{m}

\newcommand{\Mmax}{\M_{\max}}
\newcommand{\Mmin}{\M_{\min}}
\newcommand{\bmin}{b_{\min}}
\newcommand{\bmax}{b_{\max}}
\newcommand{\Bmin}{B_{\min}}
\newcommand{\Bmax}{B_{\max}}
\newcommand{\hatsj}{\hat{s}_{j}}

\newcommand{\preMmax}{\M^{(i)}_{\max}}
\newcommand{\postMmax}{\M'^{(i)}_{\max}}

\newcommand{\spred}{\hat{\mathbf{s}}}
\newcommand{\ipb}{\textsc{IPR}}
\newcommand{\ipbshort}{\textsc{IPR}}
\newcommand{\onecons}{\textsc{1-consistent}}
\newcommand{\lptpar}{\textsc{LPT-Partition}}
\newcommand{\lptrb}{\textsc{LPT-Rebalance}}
\newcommand{\ssp}{SSP}
\newcommand{\epsilonp}{\epsilon'}
\newcommand{\ipr}{\textsc{IPR}}
\newcommand{\bagsipr}{\B_{\ipb}}

\usepackage[textsize=tiny]{todonotes}
\newcommand{\INP}{\item[{\bf Input:}]}

\begin{document}
\title{
Scheduling with Speed Predictions
}
%
%
\author{Eric Balkanski\inst{1}
\and
Tingting Ou\inst{1}
\and
Clifford Stein\inst{1}
\and
Hao-Ting Wei\inst{1}
}

\authorrunning{E.Balkanski et al.}
%
\institute{Department of Industrial Engineering and Operations Research, Columbia University in the City of New York, New York, United States. \\
\email{\tt eb3224@columbia.edu}\\
\email{\tt to2372@columbia.edu}\\
\email{\tt cliff@ieor.columbia.edu}\\
\email{\tt hw2738@columbia.edu}}
\maketitle              

\begin{abstract}\vspace{-.2cm}
Algorithms with predictions is a recent framework that has been used to overcome pessimistic worst-case bounds in incomplete information settings. In the context of scheduling, very recent work has leveraged machine-learned predictions to design algorithms that achieve improved approximation ratios in settings where the processing times of the jobs are initially unknown. In this paper, we study the speed-robust scheduling problem where the speeds of the machines, instead of the processing times of the jobs, are unknown and augment this problem with predictions.

Our main result is an algorithm that achieves a $\min\{\eta^2(1+\alpha), (2 + 2/\alpha)\}$ approximation, for any $\alpha \in (0,1)$, where $\eta \geq 1$ is the prediction error. When the predictions are accurate, this approximation outperforms the best known approximation for speed-robust scheduling without predictions of $2-1/m$, where $m$ is the number of machines, while simultaneously maintaining a worst-case approximation of $2 + 2/\alpha$ even when the predictions are arbitrarily wrong. In addition, we obtain improved approximations for three special cases: equal job sizes, infinitesimal job sizes, and binary machine speeds. We also complement our algorithmic results with lower bounds. Finally, we empirically evaluate our algorithm against existing algorithms for speed-robust scheduling.
\end{abstract}
%
%
%

\vspace{-.48cm}
\section{Introduction}
\vspace{-.05cm}

In many optimization problems, the decision maker faces crucial information limitations due to the input not being completely known in advance. A natural goal in such incomplete information settings is to find solutions that have good worst-case performance over all potential input instances. However, even though worst-case analysis provides a useful measure for the robustness of an algorithm, it is also known to be a measure that often leads to needlessly pessimistic results. 

A recent, yet extensive, line of work on \emph{algorithms with predictions}  models the partial information that is often available to the decision maker and  overcomes worst-case bounds by leveraging machine-learned predictions about the inputs (see \cite{mitzenmacher2020algorithms} for a survey of the early work in this area). In this line of work,  the algorithm is given some type of prediction about the input, but the predictions are not necessarily accurate. The goal is to design  algorithms that achieve stronger bounds when the provided predictions are accurate, which are called \emph{consistency} bounds, but also maintain worst-case \emph{robustness} bounds that hold even when the predictions are inaccurate. Optimization problems that have been studied under this framework include online paging \cite{lykouris2018competitive},  scheduling \cite{KPZ18}, secretary \cite{dutting2021secretaries}, covering \cite{BMS20}, knapsack \cite{im2021online}, facility location \cite{FGGP21},  Nash social welfare \cite{banerjee2020online}, and graph \cite{azar2022online} problems. Most of the work on scheduling in this model has considered predictions about the processing times of the jobs \cite{KPZ18,MM20,LLMV20,bamas2020,im2021non}.

There is a large body of work considering uncertainty in the input to scheduling problems, including whole fields like stochastic scheduling.  Most of it studies uncertainty in the jobs. A recent line of work considers scheduling problems where 
there is uncertainty surrounding the available machines (e.g. \cite{albers2001scheduling,epstein2012universal,SZ18,EHMNSS20}). 
We emphasize the \emph{scheduling with an unknown number of parallel machines} problem, introduced by \cite{SZ18}. In this work, given a set of jobs,  there is a first {\em partitioning} stage where they must be partitioned into \emph{bags} without knowing the number of machines available and then, in a second {\em scheduling} stage, the algorithm learns the number of machines and the bags must be scheduled on the machines without being split up.

This model captures applications where partial packing decisions have to be made with only partial information about the machines. As discussed in \cite{SZ18}, such applications include MapReduce computations in shared data centers where data is partitioned into groups by a mapping function that is designed without full information about the machines that will be available in the data center, or in a warehouse where items are grouped into boxes without full information about the trucks that will be available to ship the items.  \cite{EHMNSS20} studies an extension of this model called \emph{speed-robust scheduling} where the speeds of the machines are unknown in the partitioning stage and are revealed in the scheduling stage.

In this paper, we introduce and study the problem of scheduling with machine-learned predictions about the speeds of the machines. In the two applications mentioned above, MapReduce computations and package shipping, it is natural to have some relevant historical data about the computing resources or the trucks that will be available, which can be used to obtain machine-learned predictions about these quantities. 
In the \emph{scheduling with speed predictions} problem, we are given jobs and predictions about the speeds of the $m$ machines.  In the first, \emph{partitioning stage}, jobs are partitioned into $m$ bags, using only the predictions about the speeds of the machines.  Then,  in the second, \emph{scheduling stage}, the true speeds of the machines are revealed, and the bags must be scheduled on the machines without being split up.  The goal is to use the predictions to design algorithms that achieve improved guarantees for speed-robust scheduling. The fundamental question we ask is:

\vspace{-.1cm}
\begin{center}
    \emph{Can speed predictions be used to simultaneously obtain improved guarantees for scheduling when the predictions are accurate and bounded guarantees even if the prediction errors are arbitrarily large?}
\end{center}
\vspace{-.1cm}

We focus on the classical makespan (completion time of the last completed job) minimization objective.
Two main evaluation metrics for our problem, or for any algorithms with predictions problem, are robustness and consistency. 
The consistency of an algorithm is the approximation ratio it achieves when the speed predictions are equal to the true speeds of the machines, and its robustness is its worst-case approximation ratio over all possible machine speeds, i.e., when the predictions are arbitrarily wrong. The main focus of this paper is on general job processing times and machine speeds, but we also consider multiple special cases.

Without predictions, \cite{EHMNSS20} achieves a $(2-1/m)$-approximation. Thus, if we do not trust the predictions, we can ignore them and use this algorithm to achieve a $(2-1/m)$-consistent and $(2-1/m)$-robust algorithm. On the other hand, if we fully trust the predictions, we can pretend that the predictions are correct and use a polynomial time approximation scheme (PTAS) for makespan minimization on related machines to obtain a $(1+\epsilon)$-consistent algorithm,  for any constant $\epsilon > 0$. However, as we show in Section~\ref{sec:lower-bound}, this approach would have unbounded robustness. Thus, the main challenge is to develop an algorithm that leverages predictions to improve over the best known $(2-1/m)$ approximation when the predictions are accurate, while maintaining bounded robustness guarantees even when the predictions are arbitrarily wrong.

\vspace{-.1cm}
\paragraph{Our results.} Our main result is an algorithm for minimizing makespan in the scheduling with speed predictions (SSP) model that achieves the following result, where  $\eta = \max_{i \in [m]} \frac{\max\{\hat{s}_i, s_i\}}{\min\{\hat{s}_i, s_i\}}$ is the maximum prediction error between the predicted speed $\hat{s}_i$ and the true speed $s_i$ of the $m$ machines (see Section~\ref{sec:prelim} and Appendix~\ref{sec:apperror} for addi onal discussion about the prediction error $\eta$). 

\vspace{-.1cm}
\begin{theorem*}[Theorem~\ref{thm-general} restated]
For any $\alpha \in (0,1)$, there is a deterministic $(1+\alpha)$-consistent and $(2 + 2/\alpha)$-robust algorithm for SSP with general speeds and job processing times. More generally, this algorithm achieves an approximation of $\min\{\eta^2(1+\alpha), (2 + 2/\alpha)\}$ for SSP.
\end{theorem*}
\vspace{-.1cm}


When the predictions are accurate, the $(1+\alpha)$-consistency outperforms the best-known approximation for speed-robust scheduling without predictions of $2-1/m$ \cite{EHMNSS20}, which is achieved while  maintaining a $2+ 2/\alpha$ robustness guarantee that holds  even when the predictions are arbitrarily wrong. To obtain a polynomial time algorithm,  the consistency and robustness both increase by a $1+\epsilon$ factor, for any constant $\epsilon \in (0,1)$. This $1+\epsilon$ factor is due to the PTAS algorithm for makespan minimization on related machines that we use as a subroutine.  In addition, we  obtain the following results, which are summarized in Table~\ref{tab:results}.
\vspace{-.07cm}

\begin{itemize}[leftmargin=*]
    \item For any $\alpha \in (0,1)$,  any deterministic $(1+\alpha)$-consistent algorithm has robustness  at least $1 + \frac{1-\alpha}{2\alpha} - O(\frac{1}{m})$~(\cref{LB-2-speeds}).
    If we ignore the constant factors in our result,  our algorithm matches this optimal $1/\alpha$  increase rate of the robustness.
    
    
    \item  When the job processing times are equal or infinitesimal, the best-known approximations without predictions are 1.8 and $e/(e-1)\approx 1.58$  \cite{EHMNSS20}, respectively. For these cases, our  $(1+\alpha)$-consistent algorithm achieves a robustness of $2 + 1/\alpha$~(\cref{equal-sized}) and $1 + 1/\alpha$~(\cref{Infinitesimal}), respectively.
    \item When the machine speeds are either $0$ or $1$, which corresponds to the scenario where the number of machines is unknown, the best-known approximation without predictions is $5/3$ \cite{SZ18}. We develop an algorithm that is $1$-consistent and $2$-robust (\cref{ub:1-consist}).
    We also show that, for any $\alpha \in [0,1/2)$, any deterministic $(1+\alpha)$-consistent algorithm has robustness at least $(4-2\alpha)/3$~(\cref{lb:0-1_speeds}).
     \item 
    Even when the prediction error is relatively large, our algorithm often empirically outperforms existing speed-robust algorithms that do not use predictions.  
\end{itemize}

\begin{table}[t]
\renewcommand{\arraystretch}{1.5}
\centering
\begin{tabular}{|c|c|c|c|}
\hline
Job sizes & Speeds   & Lower bound & Upper bound\\ \hline \hline
General& General   & $1 + (1-\alpha)/2\alpha- O(1/m)$ (\cref{LB-2-speeds}) 
& $2+2/\alpha$~(\cref{thm-general}) \\ \hline
Equal-size & General   & $1 + (1-\alpha)/2\alpha-O(1/m)~$(\cref{LB-2-speeds})    & $2+1/\alpha$~(\cref{equal-sized}) \\ \hline
Infinitesimal & General & $1 + (1-\alpha)^2/4\alpha-O(1/m)$ (\cref{LB-2-speeds}) 
& $1+1/\alpha$~(\cref{Infinitesimal}) \\ \hline 
General & \{0,1\} & $(4-2\alpha)/3$~(\cref{lb:0-1_speeds}) & \makecell{2  (\cref{ub:1-consist}) }\\\hline
\end{tabular}
\vspace{.25cm}
\caption{Robustness  of deterministic $(1+\alpha)$-consistent algorithms,  $\alpha \in (0, 1/2)$.}
\vspace{-.6cm}
\label{tab:results}
\end{table}

\vspace{-.3cm}

\paragraph{Technical overview.} We give an overview of the main technical ideas used to obtain our main result (\cref{thm-general}). The second stage of the SSP problem corresponds to a standard makespan minimization problem in the full information setting, so the main problem is the first stage where jobs must be partitioned into bags given predictions about the speeds of the machines. At a high level, our partitioning algorithm initially creates a partition of the jobs in bags, and a tentative assignment of the bags to machines, assuming that the predictions are the true speeds of the machines. This tentative solution is optimal if the predictions are perfect, but as we discuss in Section~\ref{sec:lower-bound}, if the predictions are wrong, its makespan may be far from optimal. To address this concern, the algorithm iteratively moves away from the initial partition in order to obtain a more robust partitioning, while also maintaining that the bags can be scheduled to give a $(1+\alpha)$-approximation of the makespan if the predictions are correct. The parameter $\alpha \in (0,1)$ is an input to the algorithm that controls the consistency-robustness trade-off, i.e., it controls how much the predictions should be trusted. We note that starting from a consistent solution and then robustifying it is a standard approach in algorithms with predictions. Our main technical contribution is in designing such a robustification algorithm for the SSP problem.


More concretely,  
let the total processing time of a bag  be the sum of the processing time of the jobs in that bag. The partitioning algorithm always maintains a tentative assignment of bags to the machines.  To robustify this assignment,   the algorithm iteratively reassigns  the bag with minimum total processing time to the machine that is assigned the bag  with maximum total processing time. If there are now $\ell$ bags assigned to this machine, we break open these $\ell$ bags, and reassign the jobs to $\ell$ new bags using the Longest Processing Time first algorithm, which will roughly  balance  the size of the $\ell$  bags assigned to this machine. Thus, at every iteration, the bags that had the maximum and minimum total processing times at the beginning of that iteration end up with approximately equal total processing times, which improves the robustness of the partition.  The algorithm terminates when the updated partition would not achieve a  $(1+\alpha)$-consistency anymore. 

The analysis of the $(2 + 2/ \alpha)$-robustness 
consists of three main lemmas.   The algorithm and analysis use a parameter $\beta$, which is the ratio of the maximum total processing time of a bag that contains at least two jobs to the minimum total processing time of a bag.  We use this particular parameter partly to handle the case of very large jobs.  Informally, both the algorithm and the adversary will need to put that one job in its own bag and on its own machine, so we can just ``ignore" such jobs.  
We first show that if we can solve the second-stage scheduling problem optimally, then the robustness achieved by any partition is at most $\max\{2, \beta\}$.  
Then, we show that at each iteration, the minimum total processing time of a bag is non-decreasing. Finally, we use this monotonicity property to show that, for the partition returned by the algorithm,  $\beta \leq 2 + 2/ \alpha$. Together with the first lemma, this implies that the algorithm achieves a $(2 + 2/ \alpha)$-robustness.  The last lemma requires a careful argument to show that, if $\beta > 2 + 2/ \alpha$, then an additional iteration of the algorithm does not break the $1+\alpha$ consistency achieved by the current partition.  To obtain a polynomial-time algorithm, we pay an extra factor of $1+\epsilon$ in the scheduling stage by using the PTAS of~\cite{HS98}.



\section{Preliminaries}
\label{sec:prelim}

We first describe the speed-robust scheduling problem, which was introduced by~\cite{EHMNSS20} and builds on the scheduling with an unknown number of machines problem from~\cite{SZ18}. There are $n$ jobs with processing times $\p = (p_1, \ldots, p_n) \geq \bm{0}$ and $m$ machines with speeds $\s = (s_1, \ldots, s_m) > \bm{0}$ such that the time needed to process job $j$ on machine $i$ is $p_j / s_i$.\footnote{The non-zero speed assumption is for ease of notation. Having a machine with speed $s_i = 0$ is equivalent to $s_i = \epsilon$ for $\epsilon$ arbitrarily small since in both cases no schedule can assign a job to  $i$ without the completion time of this job being arbitrarily large.} The  problem consists of the following two stages. In the first stage, called the partitioning stage, the speeds of the machines  are unknown and the jobs must be partitioned into $m$ (possibly empty) bags $B_1, \ldots, B_m$ such that $\cup_{i \in [m]} B_i = [n]$ (where $[n] = \{1, \ldots, n\}$) and $B_{i_1} \cap B_{i_2} = \emptyset$ for all $i_1, i_2 \in [m]$, $i_1 \neq i_2$. In the second stage, called the scheduling stage,  the speeds $\s$ are revealed to the algorithm and each bag $B_i$ created in the partitioning stage must be assigned, i.e., scheduled, on a machine without being split up.

The paper on speed-robust scheduling, \cite{EHMNSS20}, considers the classical makespan minimization objective.  Let  $\M_i$ be the collection of bags assigned to machine $i$, the goal is to minimize $\max_{i \in [m]} (\sum_{B \in \M_i} \sum_{j \in B} p_j)/ s_i$. An algorithm for speed-robust scheduling is $\beta$-robust if it achieves an approximation ratio of $\beta$ compared to the optimal schedule that knows the speeds in advance, i.e., $\max_{\p, \s} alg(\p, \s)/opt(\p, \s) \leq \beta$ where $alg(\p, \s)$ and $opt(\p, \s)$ are the makespans of the schedule returned by the algorithm (that learns $\s$ in the second stage)  and the optimal schedule (that knows $\s$ in the first stage).

We augment the speed-robust scheduling problem with predictions about the speeds of the machines and call this problem Scheduling with Speed Predictions  (\ssp). The difference between SSP and speed-robust scheduling is that, during the partitioning stage, the algorithm is now given access to, potentially incorrect, predictions $\spred = (\hat{s}_1, \ldots, \hat{s}_m) \geq 0$ about the speeds of the machines (see Appendix~\ref{sec:applearning} for additional discussion about how we learn the machine speeds and obtain $\hat{s}$). The true speeds of the machines $\s$ are revealed during the scheduling stage, as in the speed-robust scheduling problem.  We also want to minimize the   makespan.

Consistency and robustness are two standard measures in algorithms with predictions \cite{lykouris2018competitive}. An algorithm is $c$-consistent if it achieves a $c$ approximation ratio when the predictions are correct, i.e., if $\max_{\p, \s} alg(\p, \s, \s)/opt(\p, \s) \leq c$ where $alg(\p, \spred, \s)$ is the makespan of the schedule returned by the algorithm when it is given predictions $\spred$ in the first stage and speeds $\s$ in the second stage. An algorithm is $\beta$-robust if it achieves a $\beta$ approximation ratio when the predictions can be arbitrarily wrong, i.e., if $\max_{\p, \spred, \s} alg(\p, \spred, \s)/opt(\p, \s) \leq \beta$. We note that a $\beta$-robust algorithm for speed-robust scheduling  is also a  $\beta$-robust (and $\beta$-consistent) algorithm for \ssp  \ which ignores the speed predictions.

The main challenge in algorithms with predictions problems is to simultaneously achieve ``good” consistency and robustness, which requires partially trusting the predictions (for consistency), but not trusting them too much (for robustness). In particular, the goal is to obtain an algorithm that achieves a consistency that improves over the best known approximation without predictions ($2 - 1/m$ for speed-robust scheduling), ideally close to the best known approximation in the full information setting ($1+\epsilon$, for any constant $\epsilon > 0$, for makespan minimization on related machines), while also achieving bounded robustness.  

Even though consistency and robustness capture the main trade-off in SSP, we are also interested in giving approximation ratios as a function of the prediction error.  It is important, in any algorithms with predictions problem, to define the prediction error appropriately, so that it actually captures the proper notion of error in the objective.  It might seem that, for example, $L_1$ distance between the predictions and data is natural, but for many problems, including this one, such a definition would mainly give vacuous results. 
We define the prediction error $\eta \geq 1$ to be the maximum ratio\footnote{We scale $\s, \spred$ such that $\max_{i}s_i = \max_i \hat{s}_i$ before computing $\eta$, to make sure the speeds are on the same scale.} between the true speeds $\s$ and the predicted speeds $\spred$, or vice versa, i.e., $\eta(\spred, \s) = \max_{i \in [m]}\frac{\max\{\hat{s}_i, s_i\}}{\min\{\hat{s}_i, s_i\}}$. Given a bound $\eta$ on the prediction error, an algorithm achieves a $\gamma(\eta)$ approximation if  $\max_{\p, \spred, \s: \eta(\spred, \s) \leq \eta} alg(\p, \spred, \s)/opt(\p, \s) \leq \gamma(\eta)$.

 Given arbitrary bags $B_1, \ldots, B_m$, the scheduling stage corresponds to a standard makespan minimization problem in the full information setting, for which  polynomial-time approximation schemes (PTAS) are known \cite{HS98}. Thus, the main challenge is the partitioning stage. We define the  consistency and robustness of a partitioning algorithm $\A_P$ to be the consistency and robustness achieved by the two-stage algorithm that first runs $\A_P$  and then solves the scheduling stage optimally.  If we want to require that algorithms be polynomial time, we may simply run the PTAS for makespan minimization in the scheduling stage, and the bounds increase by a $(1+ \epsilon)$ factor. We will not explicitly mention this in the remainder of the paper. 

\section{Consistent Algorithms are not Robust}
\label{sec:lower-bound}

A natural first question is whether there is an algorithm with optimal consistency that also achieves a good robustness. We answer this question negatively and show that there exists an instance for which any $1$-consistent algorithm cannot be $o(n)$-robust. This impossibility result is information-theoretic and is not due to computational constraints.
The proofs in this section are deferred to Appendix~\ref{applowerbound}.

\begin{proposition}
\label{1-consist}
For any $n > m$, there is no algorithm that is $1$-consistent and $\frac{n-m+1}{\lceil n/m\rceil}$-robust, even in the case of  equal-size jobs. In particular, for $m = n / 2$, there is no algorithm that is $1$-consistent and $o(n)$-robust.
\end{proposition}

The bad instance used to achieve this result has $n$ unit-sized jobs with processing time $p_j = 1$ for $j \in [n]$ and $m < n$ machines where one machine is predicted to be much faster than the other machines, which are also predicted to have equal speed:  $\hat{s}_1  =  n - m +1$ and $\hat{s}_i =  1$ for $i \in \{2, \ldots, m\}$. The proof of Proposition~\ref{1-consist}  shows that a $1$-consistent algorithm must partitions the jobs into $m$ bags such that $m-1$ bags contain a single job and one bag contains the remaining $n - m + 1$ jobs. However, if the true machine speeds are $s_i = 1$ for all $i \in [m]$, then this partition achieves a poor robustness due to the large bag.

This result can be extended using a similar construction to show that there is a necessary non-trivial trade-off between consistency and robustness for the SSP problem. In particular, the robustness of any deterministic algorithm for SSP must grow inversely proportional as a function of the consistency.

\begin{theorem}
\label{LB-2-speeds}
    For any $\alpha \in (0, 1)$, if a deterministic algorithm for \ssp \ is $(1+\alpha)$-consistent,  then its robustness is at least $1 + \frac{1-\alpha}{2\alpha} - O(\frac{1}{m})$
    , even in the case where the jobs have equal processing times. In the special case where the processing times are infinitesimal, the robustness of a deterministic $(1+\alpha)$-consistent algorithm is at least $1 + \frac{(1-\alpha)^2}{4\alpha}-  O(\frac{1}{m}).$ 
\end{theorem}

Recall that in the setting without predictions, the best known algorithm is $(2 - 1/m)$-robust (and thus also $(2 - 1/m)$-consistent) \cite{EHMNSS20}. Since we have shown that algorithms with near-optimal consistency must have unbounded robustness, a main question is thus whether it is even possible to achieve a consistency that improves over $(2- 1/m)$ while also obtaining bounded robustness. We note that the natural idea of randomly choosing to run the $(2-1/m)$-robust algorithm or an algorithm with near-optimal consistency (with unbounded robustness), aiming to hedge between robustness and consistency, does not work since the resulting algorithm would still have unbounded robustness due to SSP being a minimization problem.


\section{The Algorithm}

In this section, we give an algorithm for scheduling with speed predictions with arbitrary-sized jobs  that achieves a $\min\{\eta^2(1+\epsilon)(1+\alpha), (1+\epsilon)(2 + 2/\alpha)\}$ approximation   for any constant $\epsilon \in (0,1)$ and any $\alpha \in (0,1)$.  


Our algorithm, called \ipb \ and formally described in \cref{alg-general}, takes as input the processing times of the jobs $\p$, the predicted speeds of the machines $\spred$, an accuracy parameter $\epsilon$, a consistency goal $1 + \alpha$, and a parameter $\rho$ that influences the ratio between the size of the smallest and largest bags.  
For general job processing times and machine speeds, we use $\rho = 4$. For some special cases in \cref{sec:specialcasesjobsizes}, we use $\rho = 2$.
\ipb \ first uses the PTAS for makespan minimization \cite{HS98} to construct a partition of the jobs into bags $B_1, \ldots, B_m$ such that scheduling the jobs in $B_i$ on machine $i$ achieves a $1+\epsilon$ approximation when the predictions are correct. In other words, it initially assumes that the predictions are correct and creates a $(1+\epsilon)$-consistent partition of the jobs into bags. In addition, it also creates a tentative assignment $\M_1 = \{B_1\}, \ldots, \M_m = \{B_m\}$ of the bags $B_1, \ldots, B_m$ on the machines.

Even though this tentative assignment achieves a good consistency, its robustness is arbitrarily poor. To improve the robustness, the main idea of our algorithm is to iteratively rebalance this partition while maintaining a bound $(1 + \epsilon)(1 + \alpha)$ on its consistency. The algorithm calls the subroutine \textsc{LPT-Rebalance} at each iteration to rebalance the bags and  modify $\M_1, \ldots, \M_m$.

We define the processing time $p(B)$ of a bag $B$ to be the sum  of the processing times of the jobs in that bag, i.e., $p(B) = \sum_{j \in B} p_j$.  The algorithm terminates either when scheduling the bags in each $\M_i$ on machine $i$ violates the desired $(1+\epsilon)(1+\alpha)$ consistency bound or when the ratio of the largest processing time of a bag containing at least two jobs to the smallest processing time of a bag is at most $\rho$. To verify the consistency bound, the algorithm compares the makespan of the new tentative assignment to the makespan $\OPTC$ of the initial assignment, assuming that the speed predictions are correct.

\begin{algorithm}[H]
    \caption{\textsc{Iterative-Partial-Rebalancing} \ (\ipbshort)}
    \label{alg-general}
\begin{algorithmic}[1]
    \INP predicted machine speeds $\hat{s}_1 \geq \cdots \geq \hat{s}_m$,  job processing times $p_1, \ldots, p_n$, consistency $1 + \alpha$, accuracy $\epsilon \in (0,1)$, maximum bag size ratio $\rho \geq 1$
    \State $\{B_1, \ldots, B_m\} \leftarrow$ a $(1+ \epsilon)$-consistent partition such that $p(B_1) \geq \cdots \geq p(B_m)$
    \State $\OPTC \leftarrow \max_{i \in [m]} p(B_i)/ \hat{s}_i$
    \State $\M_1,  \cdots, \M_{m} \leftarrow \{B_1\}, \cdots, \{B_m\}$
    \State \textbf{while}  $\max _{B \in \cup_{i } \M_i, |B| \ge 2}p(B)   > \rho  \min _{B \in \cup_{i} \M_i} p(B)$  \textbf{do} \label{li:ipr-while}
         \State \quad $\M'_1,  \ldots, \M'_{m} \leftarrow $  \lptrb ($\M_1, \ldots, \M_m$)
        \State \quad  \textbf{if} $\max_{ i \in [m]} \sum_{B \in \M'_i} p(B) / \hat{s}_i   > (1 + \alpha) \OPTC$ \textbf{then} \label{li:ipr-if}
        \State \quad \quad $\{B_1, \ldots, B_m\}  \leftarrow \cup_{i \in [m]} \M_i$
        \State \quad \quad \textbf{return}  $\{B_1, \ldots, B_m\} $
         \State \quad $\M_1,  \cdots, \M_{m} \leftarrow \M'_1,  \cdots, \M'_{m}$
        \State $\{B_1, \ldots, B_m\}  \leftarrow \cup_{i \in [m]} \M_i$
        \State\textbf{return}  $\{B_1, \ldots, B_m\} $
\end{algorithmic}
\end{algorithm}

\paragraph{The LPT-Rebalance subroutine.} This subroutine first moves the bag $B_{\min}$ with the smallest processing time to the collection of bags $\M_{\max}$ that contains the bag with the largest processing time among the bags that contain at least two jobs. Let $\ell$ be the number of bags in $\M_{\max}$, including $B_{\min}$. The subroutine then balances the processing time of the bags in $\M_{\max}$ by running the Longest Processing Time first (LPT) algorithm over all jobs in bags in $\M_{\max}$, i.e. jobs in $\cup_{B \in \M_{\max}} B$, to create $\ell$ new, balanced, bags that are placed in $\M_{\max}$. \lptrb \ finally returns the updated assignment of bags to machines $ \M_1, \ldots, \M_m$. We note that among these $m$ collections of bags, only two, $\M_{\min}$ and $\M_{\max}$, are modified. We illustrate this rebalancing with an example in Figure~\ref{fig:illustration}.

\begin{figure*}[t]
    \centering
    \includegraphics[width=12cm]{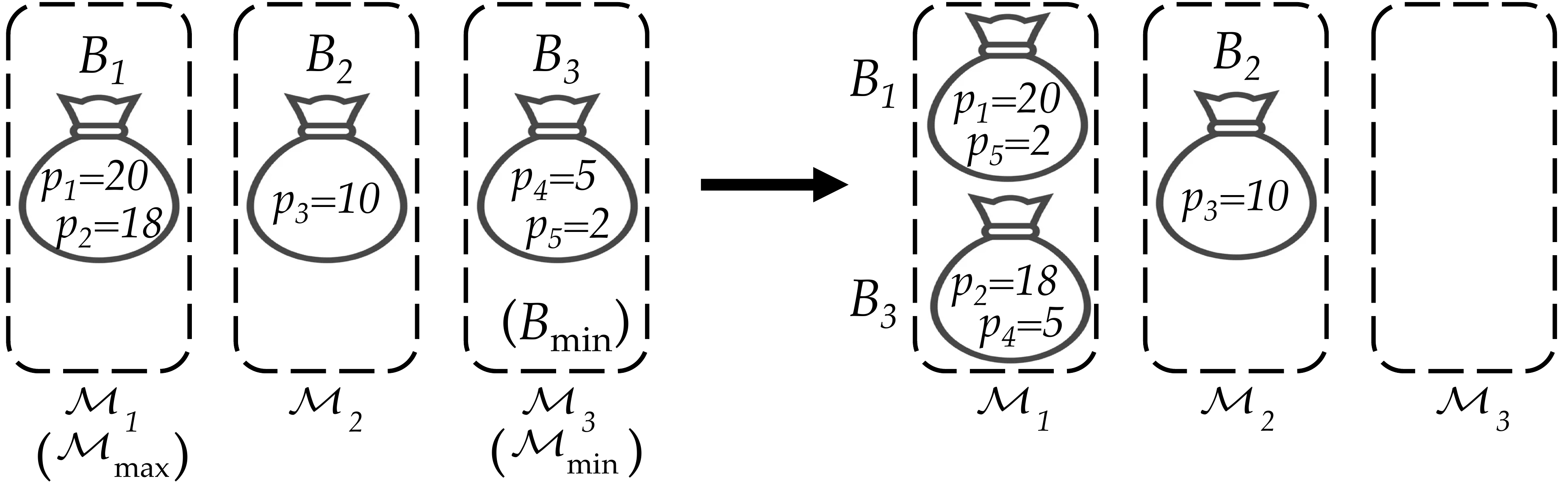}
    \caption{One iteration of the \ipb \ algorithm on an example with $m=3$ bags and $n=5$ jobs.} 
    \label{fig:illustration}
\end{figure*}

\begin{algorithm}[H]
    \caption{\lptrb}
    \label{sub-LPT}
\begin{algorithmic}[1]
    \INP assignment of bags to machines $ \M_1, \ldots, \M_m$
    \State $B_{\min} \leftarrow \argmin _{B \in \cup_i \M_i}p(B)$ 
    \State $\M_{\min} \leftarrow $ the collection of bags $\M$ such that $B_{\min} \in \M$
    \State $\M_{\max} \leftarrow \argmax_{\M_i : i \in [m]} \max_{B \in \M_i : |B| \geq 2} p (B)  $
    \State $\M_{\max} \leftarrow \M_{\max} \cup \{B_{\min}\};  \M_{\min} \leftarrow \M_{\min} \setminus \{B_{\min}\}$
    \State $J_{\max} \leftarrow \cup_{B \in \M_{\max}} B; \ell \leftarrow |\M_{\max}|$
    \State $B'_1, \ldots, B'_{\ell} \leftarrow \{\}, \ldots, \{\}$
    \State \textbf{while} $|J_{\max}| > 0$ \textbf{do}
    \State \quad $j^{\prime} \leftarrow \argmax_{j \in J_{\max}} p_j$
    \State \quad $B^{\prime} \leftarrow \argmin_{B \in \{B'_1, \ldots, B'_{\ell }\}} p(B)$
    \State \quad $B^{\prime} \leftarrow B^{\prime} \cup \{j^{\prime}\}; J_{\max} \leftarrow J_{\max} \setminus \{j^{\prime}\}$
    \State $\M_{\max} \leftarrow \{B'_1, \ldots, B'_{\ell}\}$
    \State \textbf{return} $\M_1, \ldots, \M_m$
\end{algorithmic}
\end{algorithm}

\subsection{Analysis of the algorithm}
\label{sec:alg-performance}

We first show that \ipb \ with parameter $\rho = 4$ in the general case is a $(1+\epsilon)(1+ \alpha)$-consistent and $(2+2/\alpha)$-robust partitioning algorithm (\cref{lem:consistency} and \ref{lem:robustness}). Then, we use these consistency and robustness guarantees to obtain the $\min\{\eta^2(1+\epsilon)(1+\alpha), (1+\epsilon)(2 + 2/\alpha)\}$ approximation as a function of the prediction error $\eta$ (\cref{thm-general}). Finally, we analyze the running time (\cref{runtime}). The main challenge is to analyze \ipb's robustness.

\subsubsection{Analysis of the algorithm's consistency and robustness} 

The consistency of the algorithm is almost immediate.


\begin{lemma}
\label{lem:consistency}
     For any constant $\epsilon \in (0,1)$ and any $\alpha \in (0,1)$, \ipb \ with any $\rho \geq 1$ is a $(1+\epsilon)(1+ \alpha)$-consistent partitioning algorithm.
\end{lemma}

\begin{proof}
Consider the final tentative assignment of the bags on the machines $\M_1, \ldots, \M_m$ when \ipb \ terminates. With true speeds $\s$, the makespan of this schedule  is $\max_{ i \in [m]} \sum_{B \in \M_i} p(B) /s_i$. When the speed predictions are correct, i.e., $\s = \spred$, we have
\begin{align*}
\max_{ i \in [m]} \frac{\sum_{B \in \M_i} p(B)}{s_i}  & = \max_{ i \in [m]} \frac{\sum_{B \in \M_i} p(B)}{\hat{s}_i}  \leq (1+\alpha) \OPTC \leq (1+\alpha)(1+\epsilon) opt(\p, \s) .
\end{align*}

Line~\ref{li:ipr-if} of \ipb \ enforces the first inequality. 
For the second inequality, observe that  when $\s = \spred$, $\OPTC$ is the makespan of the initial assignment, which is a $1+\epsilon$ approximation to the optimal makespan $opt(\p, \s)$. Since there exists an assignment of the bags returned by \ipb \ that achieves a $(1+\epsilon)(1+\alpha)$ approximation when $\s = \spred$, \ipb \ is a $(1+\epsilon)(1+\alpha)$-consistent partitioning algorithm.
\end{proof}

In the remainder of this section, we analyze the robustness of the $\ipb$ algorithm. First, we show that the ratio $\beta(\B) = \frac{\max_{B \in \B, |B| \ge 2}p(B)}{  \min_{B \in \B} p(B)}$  of the maximum total processing time of a bag containing at least two jobs to the minimum total processing time of a bag  can be used to bound the robustness of any partition $\B$. 


\begin{lemma}
\label{lem:robustnessmainlemma}
    Let $\B = \{ B_1, \cdots, B_m \}$ be a partition of $n$ jobs with processing times $p_1, \ldots p_n$ into $m$ bags. Then $\B$ is a $\max\{2, \beta(\B) \}$-robust partition, where $\beta(\B) = \frac{\max_{B \in \B, |B| \ge 2}p(B)}{  \min_{B \in \B} p(B)}$.
\end{lemma}

 This lemma generalizes Theorem~3.3 in \cite{EHMNSS20} which shows a robustness bound of 2 when $\beta (\B) \leq 2$ and its proof is deferred to Appendix~\ref{sec:appanalysis}. 
 We let $\bagsipr^{(i)}$ denote the collection of all bags $B$ at iteration $i$ of the \ipb~algorithm (\cref{alg-general}) and define $b^{(i)}_{\min} = \min_{B \in \bagsipr^{(i)}} p(B)$ to be the minimum processing time of a bag at each iteration $i$. 
 In the next lemma, we show that $b^{(i)}_{\min}$ is non-decreasing in $i$.

\begin{lemma}
\label{lem:non-decre}
    At each iteration $i$ of \ipb \ with $\rho = 4$, $b_{\min}^{(i+1)} \geq b_{\min}^{(i)}$. 
\end{lemma}

Before proving Lemma~\ref{lem:non-decre}, we first state two lemmas that are needed for the proof. The first is a useful property of the LPT algorithm that was shown in \cite{EHMNSS20}. 

\begin{lemma}\cite{EHMNSS20}
\label{lem-LPT}
For any job processing times $p_1, \ldots, p_n$ and number of machines $m$, the partition $\B_{\textsc{LPT}} = \{B_1, \cdots, B_m\}$ returned by the LPT algorithm on these jobs  satisfies $\beta(\B_{\textsc{LPT}} ) \leq 2$. 
\end{lemma}

The second, whose proof is also deferred to Appendix~\ref{sec:appanalysis}, bounds the minimum total processing time of a bag if the maximum total processing time of a non-singleton bag is at most twice as large as the minimum total processing time of a bag in a partition. 

\begin{lemma}
\label{lem-simple}
For any job processing times $p_1, \ldots, p_n$ and partition $\B = \{B_1, \cdots, B_m\}$ of the jobs, if  $\beta(\B)  \leq 2$, then $\min_{B \in \B} p(B) \geq \frac{\sum_{j=1}^n p_j}{2m-1}.$
\end{lemma}


We are now ready to prove Lemma~\ref{lem:non-decre}.

\begin{proof}[Proof of Lemma~\ref{lem:non-decre}]
Let $\preMmax$ and $\postMmax$ denote $\M_{\max}$ in the $i^{th}$ iteration of the \ipb \ Algorithm (Line~\ref{li:ipr-while}) before and after we add the bag $\Bmin^{(i)}$ to it  and balance it. Let $\ell_i$ be the number of bags in $\postMmax$, which means that $\preMmax$ has $(\ell_i-1)$ bags before receiving $\Bmin^{(i)}$.  Let $b^- = \min_{B \in \M_{\max}^{(i)}} p(B)$ be the minimum processing time of a bag in $\preMmax$. Let $\C^+ = \{ B \in \postMmax : |B| = 1, \ p(B) > b^-\}$ and $\C^- = \{ B \in \postMmax: |B| = 1, \  \max _{B \in \postMmax, |B| \ge 2}p(B) < p(B) \leq b^- \}$.

Our main goal is to prove that $\min_{B\in \postMmax} p(B) \leq \bmin^{(i)}$ which implies Lemma~\ref{lem:non-decre}. We first argue that for any singleton bag $B \in \C^+$, we also have that $B \in \preMmax$. If $\preMmax$ contains only one bag, then the statement is trivially true because $b^-$ is the total processing of the only bag in $\preMmax$, and we add a bag of total processing time $\bmin^{(i)} \leq b^-$ to $\preMmax$, so in $\postMmax$ there does not exist a job of processing time larger than $b^-$.
If $\preMmax$ has at least two bags, then $\preMmax$ has been balanced in some previous iteration.
Assume for the sake of contradiction that before we run the $\lptrb$ subroutine, the job $j$ is not in a singleton bag and the bag that contains $j$ also contains another job $k$. 
We consider the LPT process that produces the bags in $\preMmax$. If $p_j > p_k$, then we assign job $j$ prior to job $k$. When we assign job $k$, the bag with $p_j$ inside has processing time more than the bag of processing time $b^-$, so we would not place the job $k$ into the same bag as job $j$, a contradiction. If $p_k \geq p_j$, then again when we assign job $j$, the bag that contains job $k$ has processing time at least $p_j > b^-$ so we would not place job $j$ into the same bag as job $k$, contradiction. 

Since we have that $B \in \preMmax$ and 
$p(B) > b^- > \bmin^{(i)}$ for all $B \in \C^+$, we focus on the bags  $\preMmax \setminus \C^+$. 
Consider the total processing time of the bags in $\preMmax \setminus \C^+$ before $\lptrb$ is executed in the $i^{th}$ iteration. These jobs in $\preMmax \setminus \C^+$ are contained in $(\ell_i - 1 - |\C^+|)$ bags. Recall that $b^-$ is the minimum processing time of a bag in $\Mmax^{(i)}$, so we have
\begin{align}
    \sum_{B \in \preMmax \setminus \C^+}p(B) \geq (\ell_i -1 -|\C^+|)b^-. \label{ineq-initial-gpsize-bound}
\end{align}

After we add $B^{(i)}_{\min}$ to $\preMmax$, the total processing time of this collection of bags excluding the jobs in $\C^+$ is $\sum_{B \in \postMmax \setminus \C^+} p(B) = \sum_{B \in \preMmax \setminus \C^+}p(B) +  \bmin^{(i)} $.  Since the singleton bags $\C^-$ have total processing time at most $|\C^-|b^-$, we then have that 
$\sum_{B \in \postMmax \setminus \{\C^+ \cup \C^-\}} p(B) \geq \sum_{B \in \preMmax \setminus \C^+}p(B) +  \bmin^{(i)} - |\C^-|b^-$. By Lemma~\ref{lem-LPT}, for  bags $B, B' \in \postMmax \setminus \{\C^+ \cup \C^-\}$, we have $p(B) \leq 2p(B')$. Therefore, by Lemma~\ref{lem-simple}, we have that 
\begin{align}
    \min_{B\in \postMmax \setminus \{\C^+ \cup \C^-\}} p(B)  
    & \geq \frac{\sum_{B \in \postMmax \setminus \{\C^+ \cup \C^-\}} p(B)}{2|\postMmax \setminus \{\C^+ \cup \C^-\}| - 1} \\
    & \geq \frac{\sum_{B \in \preMmax \setminus \C^+}p(B) - |\C^-|b^- +  \bmin^{(i)} }{2 (\ell_i - |\C^+| - |\C^-| ) - 1}. \label{ineq-minbags-bound}
\end{align}

Next, note that  $\max _{B \in \cup_{i} \M_i, |B| \ge 2}p(B) > 4\bmin^{(i)}$ by the algorithm with $\rho = 4$. We also have that $\max _{B \in \cup_{i} \M_i, |B| \ge 2}p(B) \leq 2b^-$ by Lemma~\ref{lem-LPT}, which implies that  $b^- > 2 \bmin^{(i)}$.  Combining the inequalities (\ref{ineq-initial-gpsize-bound}) and (\ref{ineq-minbags-bound}), we obtain
    \begin{align*}
        \min_{B\in \postMmax} p(B) &  = \min_{B\in \postMmax \setminus \{\C^+ \cup \C^-\}} p(B)   \\
        & \geq \frac{\sum_{B \in \preMmax \setminus \C^+} p(B) - |\C^-| b^- + \bmin^{(i)}}{2 (\ell_i - |\C^+| - |\C^-|) - 1 }\\
        &\geq \frac{(\ell_i - 1 - |\C^+|)b^{-} - |\C^-| b^{-} + \bmin^{(i)}}{2 (\ell_i -|\C^+| - |\C^-|) - 1}\\
        & > \frac{(\ell_i - |\C^+|- |\C^-|)2\bmin^{(i)} - \bmin^{(i)} }{2(\ell_i -|\C^+|- |\C^-|) - 1}\\
        & = \bmin^{(i)}.
    \end{align*}

We conclude that 
\begin{align*}
b^{(i+1)}_{\min} \geq \min\left\{\min_{B \in \M^{'(i)}_{\max}}p(B), b^{(i)}_{\min}\right\} &= b^{(i)}_{\min}. 
\end{align*}
   
\end{proof}

In the remainder of the section, for ease of notation, we let $\beta = \beta(\bagsipr)$, where $\bagsipr$ is the partition returned by \ipb \, and bound the value of $\beta$. Similarly, let $b_{\min} = \min_{B \in \bagsipr} p(B)$ be the minimum processing time of a bag returned by \ipb. 
Additionally, we let all quantities such as $\M_i$ and $B_i$ refer to the quantities $\M_i$ and $B_i$ when the algorithm terminates, unless noted otherwise.  We let $\bmax = \max _{B \in \bagsipr, |B| \ge 2}p(B)$, $\Bmax$ be the bag with processing time $\bmax$, $W = \sum_{B\in \Mmax}{p(B)}$ be the total processing time of $\Mmax$ and $\ell = |\Mmax|$. 

We note that $\bmin > 0$ when the algorithm terminates, because all the empty bags created in the initial partition will be moved to some collection of bags and eliminated by the \lptrb \ subroutine in the first  iterations. If the algorithm has $\bmin^{(i)} = 0$ in some iteration $i$, then it will keep running because $\bmax^{(i)}/\bmin^{(i)} = \infty > \rho$ and that the consistency bound won't be broken as we are adding an empty bag to $\Mmax$.
Therefore, the algorithm would never end up with $\bmin = 0$.

In the next lemma, given an upper bound of $b_{\max}$, we bound the ratio $\beta$ assuming that the minimum total processing of a bag in each iteration is nondecreasing. Note that $\beta$ is also equal to $\bmax/\bmin$ using the newly introduced notations. 

\begin{lemma} \label{sizeratio}
    Let $\bagsipr = \{B_1, \ldots, B_m\}$ be the partition of the $n$ jobs returned by \ipb \ with $\rho \geq 1$. Assume that at each iteration $i$ of \ipb, $b_{\min}^{(i+1)} \geq b_{\min}^{(i)}$. Let $\M_j$ be the collection of bags such that $\M_j = \Mmax$. 
    If $\bmax \leq c_1 W + c_2$ for $c_1,c_2 \geq 0$, then
    $\beta \leq \max\{\rho, c_1\left(\ell-1+\frac{\ell}{\alpha}\right) + \frac{c_2{\ell}}{\alpha\hat{s}_{j}\OPTC}\}.$ 
\end{lemma}

\begin{proof}
 If $\bmax / \bmin < \rho$ then the algorithm terminates and satisfies the condition $\beta = \bmax / \bmin \le \max\{\rho, c_1\left((\ell-1)+\frac{\ell}{\alpha}\right) + \frac{c_2{\ell}}{\alpha\hat{s}_{j}\OPTC}\}$. Therefore, we can assume that $\bmax / \bmin > \rho$ and the algorithm terminates due to the violation of the consistency bound.

Since $\bmax \leq c_1 W + c_2$, we have the following bound on $\beta$:
\begin{equation}\label{ineq-beta-bound-1}
    \beta  = \frac{\bmax}{\bmin} \leq \frac{c_1 W + c_2}{\bmin}.
\end{equation}

Note that $\M_j$ is the collection of bags such that $\M_j = \Mmax$. 
The initial processing time of $\Mmax$ when it only contained a bag $B_j$ from the $(1+\epsilon)$-consistent partition, is at most $\hatsj \OPTC$. We use $\omega$ to denote the total processing time that is added to $\Mmax$ through the execution of \ipb, so we have $W \leq \hatsj \OPTC + \omega$. Then from inequality~(\ref{ineq-beta-bound-1}), we have the following bound on $\beta$:
\begin{equation}\label{ineq-general-2}
   \beta \leq \frac{c_1( {\hat{s}}_{j}\OPTC + \omega)+c_2}{\bmin }.
\end{equation}

By the assumption that  at each iteration $i$ of \ipb, $b_{\min}^{(i+1)} \geq b_{\min}^{(i)}$, any bag previously added to $\Mmax$ must have processing time at most $\bmin$. Thus, 
\begin{align}\label{ineq-omega-ub}
 \omega \leq ({\ell}-1) \bmin.
\end{align}

To finish the last part of the analysis, we split into two cases: ($a$) $\omega \geq \frac{(\ell-1)\alpha}{\ell}\hat{s}_{j}\OPTC$, and ($b$) $\omega < \frac{(\ell-1)\alpha}{\ell}\hat{s}_{j}\OPTC$. For case ($a$), combining inequalities (\ref{ineq-general-2}) and (\ref{ineq-omega-ub}), we have:
\begin{align*}
    \beta & \leq \frac{c_1({\hat{s}}_{j}\OPTC + ({\ell} - 1) \bmin )+c_2}{\bmin}  =  \frac{c_1{\hat{s}}_{j} \OPTC }{\bmin } + c_1(\ell - 1) + \frac{c_2}{\bmin}.
\end{align*}   

From the condition of this case, we have $\bmin \geq \frac{\omega}{\ell-1} \geq \frac{\frac{(\ell-1)\alpha}{\ell}\hat{s}_{j}\OPTC}{\ell-1} = \frac{\alpha\hat{s}_{j}\OPTC}{\ell}$. Thus, $\frac{c_1{\hat{s}}_{j} \OPTC }{\bmin} \leq \frac{c_1{\ell}}{\alpha}$ and $\frac{c_2}{\bmin} \leq \frac{c_2{\ell}}{\alpha\hat{s}_{j}\OPTC}$. Then, we have:
\begin{align*}
    \beta & \leq \frac{c_1{\hat{s}}_{j} \OPTC }{\bmin } + c_1(\ell - 1) + \frac{c_2}{\bmin}  \\
    & \leq \frac{c_1{\ell}}{\alpha} + c_1(\ell - 1) + \frac{c_2{\ell}}{\alpha\hat{s}_{j}\OPTC} \\
    & = c_1\left(\ell-1+ \frac{\ell}{\alpha}\right) + \frac{c_2{\ell}}{\alpha\hat{s}_{j}\OPTC}.
\end{align*}

For case ($b$), since  adding a bag with processing time $\bmin$ to $\Mmax$ caused the algorithm to terminate, we will show that $\bmin \geq (\alpha/{\ell}){\hat{s}}_{j}\OPTC$. For the sake of contradiction, we assume the otherwise, i.e. $\bmin < (\alpha/{\ell}){\hat{s}}_{j}\OPTC$. 
Then we consider the total processing time of $\Mmax$ when the algorithm terminates. Initially the total processing time of $\Mmax$ is the total processing time of $B_j$, which is at most $\hat{s}_{j}\OPTC$. 
From the assumption of case ($b$), the total processing time of $\Mmax$ when the algorithm terminates is at most ${\hat{s}}_{j}\OPTC + \omega \leq (1 + \frac{({\ell}-1)\alpha}{{\ell}}) {\hat{s}}_{j}\OPTC$. 
Thus if $\bmin < (\alpha/{\ell}){\hat{s}}_{j}\OPTC$, after we add $\Bmin$ to $\Mmax$, the updated collection of bags $\M'_{\max} = \Mmax \cup \Bmin$ has a total processing time at most $(1 + \frac{({\ell}-1)\alpha}{{\ell}}) {\hat{s}}_{j}\OPTC + (\alpha/{\ell}){\hat{s}}_{j}\OPTC \leq (1+\alpha) {\hat{s}}_{j}\OPTC$. 
Thus, we have $\sum_{B \in \M'_{\max}}p(B) / \hat{s}_j \leq (1+\alpha) \OPTC$. For any other collection of bags, the total processing time of the jobs in such collection of bags cannot increase compared to the previous iteration. Thus, if we add $\Bmin$ to $\Mmax$, run the $\lptrb$ subroutine and let the new assignments be $\M'_1, \ldots, \M'_m$, we have $\max_{ i \in [m]} \sum_{B \in \M'_i} p(B) / \hat{s}_i \leq (1 + \alpha) \OPTC$ and the algorithm would have run another iteration. Contradiction with the fact that we stopped the algorithm. Then, we have:
\begin{align*}
    \beta &\leq \frac{c_1({\hat{s}}_{j}\OPTC + \frac{({\ell}-1)\alpha}{{\ell}}{\hat{s}}_{j} \OPTC )+c_2}{\bmin}\\ 
    &= \left(1+\frac{(\ell-1)\alpha}{\ell} \right) \frac{c_1{\hat{s}}_{j} \OPTC }{ \bmin} + \frac{c_2}{\bmin} \\
    &\leq \left(1+\frac{(\ell-1)\alpha}{\ell} \right) \frac{c_1{\hat{s}}_{j}\OPTC}{ (\alpha/{\ell}) {\hat{s}}_{j} \OPTC} + \frac{c_2}{(\alpha/{\ell}) {\hat{s}}_{j} \OPTC} \\ 
    &= \left(1+\frac{(\ell-1)\alpha}{\ell} \right) \frac{c_1 \ell}{\alpha} + \frac{c_2\ell}{\alpha{\hat{s}}_{j} \OPTC} \\
    &= c_1\left(\ell-1+ \frac{\ell}{\alpha}\right) + \frac{c_2{\ell}}{\alpha\hat{s}_{j}\OPTC}.
\end{align*}

The first inequality uses inequality~(\ref{ineq-general-2}) and the condition of this case.
The second inequality holds because we have proved that $\bmin \geq \frac{\alpha }{{\ell}}{\hat{s}}_{j}\OPTC$. 
\end{proof}

The previous lemma shows that if we can upper bound $b_{\max}$ then we can upper bound the ratio $\beta$.
In the next lemma, we prove an upper bound of $b_{\max}$.
\begin{lemma}\label{bmax}
Let $\bagsipr = \{B_1, \ldots, B_m\}$ be the partition of the $n$ jobs returned by \ipb \ with $ \rho = 4 $. We have $\bmax \leq \frac{2W}{{\ell}+1}$.
\end{lemma}

\begin{proof}
The proof is by contradiction. 
If $\bmax > \frac{2{W}}{{\ell}+1}$ then by Lemma~\ref{lem-LPT}, $\forall B\in \Mmax, p(B) \leq \bmax$ we have $p(B) >  \frac{W}{{\ell}+1}$. 
Consider the total processing time of all bags in $\Mmax$ if $\bmax > \frac{2{W}}{{\ell}+1}$. We have $\sum_{B\in \Mmax}{p(B)} >  \frac{2{W}}{{\ell}+1} + (\ell - 1) \frac{W}{{\ell}+1} > W$, which leads to a contradiction. 
\end{proof}

We are now ready to show the algorithm's robustness.


\begin{lemma}\label{lem:robustness}
For any constant $ \epsilon \in (0,1)$ and any $\alpha \in (0,1)$, \ipb \ with $\rho = 4$ is a $(2 + 2/\alpha)$-robust partitioning algorithm.
\end{lemma}

\begin{proof}
For the robustness, we first note that if $\rho = 4$, by Lemma~\ref{lem:non-decre} we have that at each iteration $i$ of \ipb, $b_{\min}^{(i+1)} \geq b_{\min}^{(i)}$. Also, by Lemma~\ref{bmax} we know that $b_{\max} \geq \frac{2}{l+1}W$ when $\rho = 4$. Therefore, we can apply Lemma~\ref{sizeratio} with $\rho = 4$, $c_1 = \frac{2}{\ell+1}$ and $c_2 = 0$, and we have that 
$\beta \leq \max\{4, \frac{2(\ell-1)}{\ell+1} + \frac{2\ell  }{\alpha(\ell+1)}\} \le 2+\frac{2}{\alpha}$. Thus, by Lemma~\ref{lem:robustnessmainlemma}, the robustness of the returned partition by \ipb \ with $\rho =4$ is $(2+2/\alpha)$.
\end{proof}
The consistency-robustness trade-off is shown in \cref{fig:tradeoff}.
We note that \cref{lem:robustness}, together with the lower bound (\cref{LB-2-speeds}), implies that,  ignoring constant factors, \ipb \ achieves the optimal $1/\alpha$ rate of increase of the robustness.

\begin{figure}[t]
 	\centering
 	\subfigure{\includegraphics[width=0.49\textwidth]{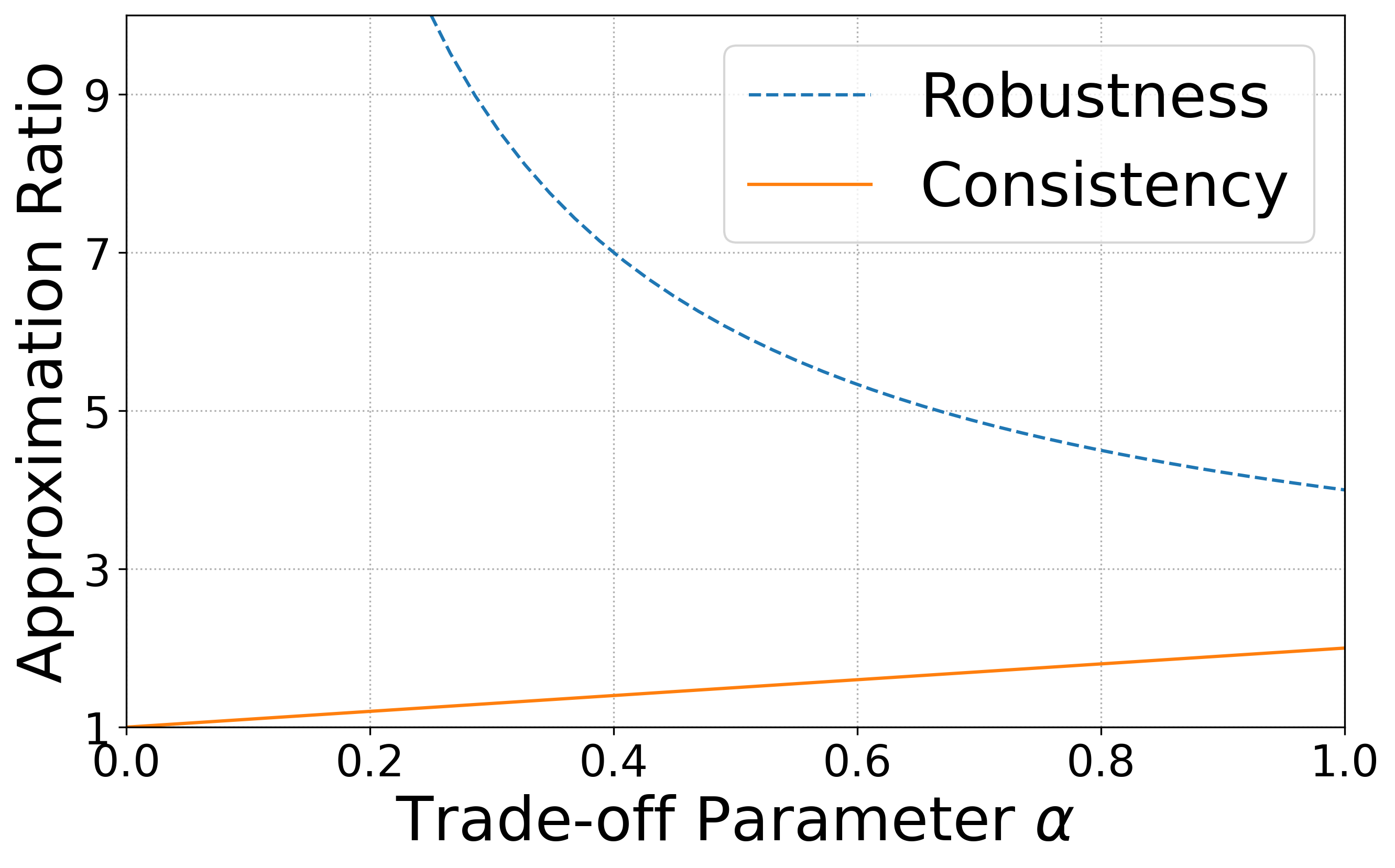}}
    \subfigure{\includegraphics[width=0.49\textwidth]{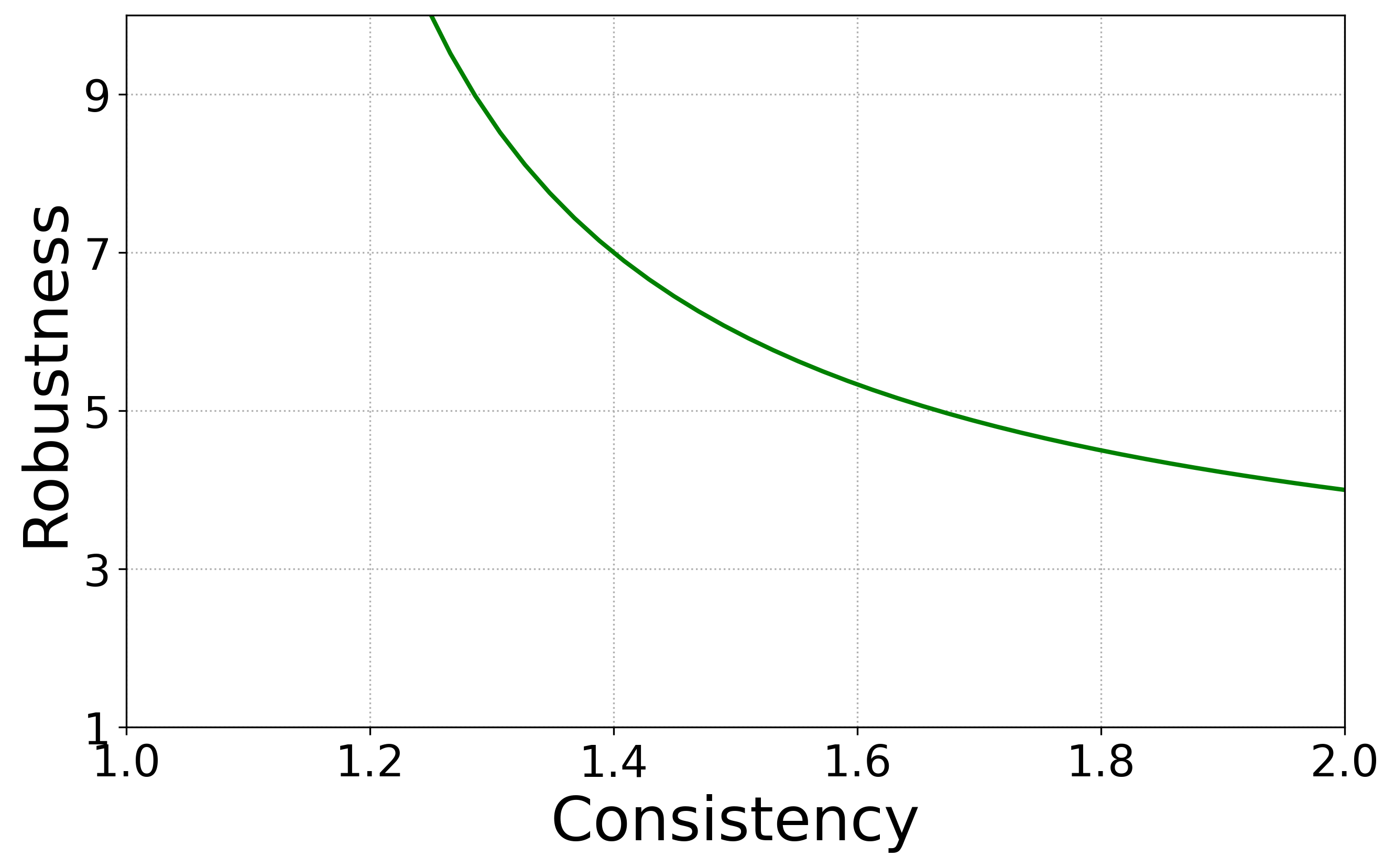}} 
    \caption{The consistency and robustness guarantee of the $\ipb$ algorithm under different parameter $\alpha$ (left) and the trade-off between consistency and robustness of the $\ipb$ algorithm (right).}
    \label{fig:tradeoff}
 \end{figure}

\subsubsection{Analysis of the algorithm's approximation as a function of the prediction error}

We extend the  consistency and robustness results for \ipb \ to obtain our main result. We show that for the SSP problem, the algorithm that runs \ipb \ in the partitioning stage and then a PTAS in the scheduling stage achieves an approximation ratio that gracefully degrades as a function of the prediction error $\eta$ from $(1+\epsilon)(1+\alpha)$ to $(1+\epsilon)(2 + 2/\alpha)$. 
The proof follows from the results for consistency (\cref{lem:consistency}) and robustness (\cref{lem:robustness}) to obtain the $\eta^2(1+\epsilon)(1+\alpha)$ and  $(1+\epsilon)(2 + 2/\alpha)$ bounds on the approximation. We give its proof in Appendix~\ref{sec:appanalysis}.


\begin{theorem}
\label{thm-general}
Consider the algorithm that runs \ipb \ with $\rho = 4$ in the partitioning stage and a PTAS for makespan minimization in the scheduling stage.  For any constant $\epsilon \in (0,1)$ and any $\alpha \in (0,1)$, this algorithm achieves a 
$\min\{\eta^2(1+\epsilon)(1+\alpha), (1+\epsilon)(2 + 2/\alpha)\}$ approximation for SSP where $\eta = \max_{i \in [m]} \frac{\max\{\hat{s}_i, s_i\}}{\min\{\hat{s}_i, s_i\}}$ is the prediction error.
\end{theorem}

If we do not care about the computation runtime; that is, we can solve each scheduling problem optimally including the initial step of \ipb \ in the partition stage and the scheduling stage, then our result improves to a $\min\{\eta^2(1+\alpha), (2 + 2/\alpha)\}$ approximation.


\subsubsection{Analysis of the algorithm's running time} 

In Appendix~\ref{sec:appanalysis}, we show that the main algorithm performs $O(m^2)$ iterations, which implies that its running time is polynomial in $n$ and $m$.

\begin{lemma}
\label{runtime}
    At most $O(m^2)$ iterations are needed for \ipb \ with $\rho = 4$ to terminate.
\end{lemma}

\section{Special Cases}
\label{sec:special-cases}

In this section, we consider two different special cases where all job processing times are either equal or infinitesimal (\cref{sec:specialcasesjobsizes}), or all machines have speeds either 0 or 1 (\cref{sec:specialcase0-1speeds}).  These special cases were considered without predictions in \cite{EHMNSS20} and \cite{SZ18}.

In the first special case where all job processing times are either equal or infinitesimal, we prove that the $\ipb$ algorithm with $\rho = 2$ achieves a better robustness.
In the second special case where the machine speeds are in $\{ 0,1 \}$, we no longer use the $\ipb$ algorithm. Instead, we propose a new partitioning algorithm that is $(1+\epsilon)$-consistent and $2(1+\epsilon)$-robust in this special case. We also prove that for any $\alpha \in [0,1/2)$, any deterministic $(1+\alpha)$-consistent algorithm has robustness at least $(4-2\alpha)/3$.

\subsection{Improved trade-off by IPR for special job processing times}
\label{sec:specialcasesjobsizes}

In this section, we show improved bounds on the robustness of \ipb \ for the special cases where the jobs either have equal processing time or are infinitesimally small. To obtain these improved bounds, we use \ipb \ with parameter $\rho = 2$ instead of $\rho = 4$.  

\subsubsection{Equal-size jobs}

In the case where jobs have equal processing times, i.e. $p_j = 1$ for all $j\in [n]$, we obtain an improved robustness of $(2+ 1/\alpha)$ for the returned partition by \ipb \ with $\rho = 2$.

\begin{theorem}
\label{equal-sized} If $p_j = 1$ for all $j \in [n]$, then, for any constant $\epsilon \in (0,1)$ and any $\alpha \in (0,1)$, \ipb \ with $\rho = 2$ is a $(1+\epsilon)(1+\alpha)$-consistent and $(2+1/\alpha)$-robust partitioning algorithm.
\end{theorem}

The proof of \cref{equal-sized} relies on the following two lemmas. The first lemma, Lemma~\ref{lem:nondecre-eq}, shows the monotonicity property for $\bmin^{(i)}$ when $\rho = 2$ in the equal-size jobs case.

\begin{lemma}
\label{lem:nondecre-eq}
    Assume all jobs have the same total processing time, i.e. $p_j = 1$ for all $j \in [n]$. At each iteration $i$ of \ipb \ with $\rho = 2$, $\bmin^{(i+1)} \geq \bmin^{(i)}$.
\end{lemma}

\begin{proof}
In this proof, we extensively use the notations introduced in the proof of \cref{lem:non-decre}. Consider the $i^{th}$ iteration of the \ipb~algorithm  (Line~\ref{li:ipr-while}). In this special case where all jobs have total processing time 1, the \lptrb \ subroutine will partition all of the jobs into bags such that the difference of the total processing time between any two bags is at most 1, i.e. $\forall B, B' \in \preMmax$ we have $p(B) \leq  p(B') + 1$. Note that we have $\bmax^{(i)} > 2 \bmin^{(i)}$, otherwise the algorithm will not execute the $(i+1)^{th}$ iteration. Let $\ell_i$ be the number of bags in $\preMmax$ before we add $\Bmin^{(i)}$. 

The sum of the total processing time of bags in $\preMmax$ after we add $\Bmin^{(i)}$ is:
\begin{align*}
    \sum_{B\in \preMmax} p(B) + \bmin^{(i)} & >  (\ell_i - 1) (2\bmin^{(i)}-1) + \bmax^{(i)} + \bmin^{(i)} > (2\ell_i + 1)\bmin^{(i)} - (\ell_i - 1).
\end{align*}

Next, we argue that the minimum processing time of a bag in $\Mmax^{(i)}$ after balancing is at least $\bmin^{(i)}$, i.e. $\delta \geq \bmin^{(i)}$. If all jobs have the same total processing time, $\lptrb$ will simply assign the $j^{th}$ job to the ($j$ mod $(\ell_i + 1)$)$^{th}$ bag. If we need all the $(\ell_i + 1)$ LPT-rebalanced bags to have total processing time at least $\bmin^{(i)}$, we only need to make sure there are at least $\bmin^{(i)}(\ell_i + 1)$ jobs as input. Since $\bmin^{(i)} \geq 1$ in this special case, we have:
\begin{align*}
    \bmin^{(i)}(\ell_i + 1) & = (2\ell_i + 1)\bmin^{(i)} - \ell_i \bmin^{(i)}\\
    & \leq (2\ell_i + 1)\bmin^{(i)} - \ell_i \\
    & \leq (2\ell_i + 1)\bmin^{(i)} - (\ell_i - 1) \\
    & < \sum_{B\in \preMmax}p(B) + \bmin^{(i)}
\end{align*}
which shows that $\delta \geq \bmin^{(i)}$. Thus, we have $\bmin^{(i+1)} \geq \min \{ \bmin^{(i)}, \delta \} = \bmin^{(i)}$. 
\end{proof}

The second lemma, Lemma~\ref{bmax-equalsized}, gives an upper bound of $b_{\max}$. 

\begin{lemma}\label{bmax-equalsized}
Let $\bagsipr = \{B_1, \ldots, B_m\}$ be the partition of the $n$ jobs returned by \ipb \ with $ \rho = 2 $. Then $\bmax \leq \frac{W}{\ell}+1$ in the case of equal-size jobs.
\end{lemma}

\begin{proof}
Since all of the jobs have processing time 1, the balancing process that generated $\Mmax$ simply distributed $W$ unit-length jobs into $\ell$ bags evenly, so $\bmax \leq \ceil{\frac{{W}}{{\ell}}} \leq W/\ell + 1$. 
\end{proof}

Combining these two lemmas, we can apply Lemma~\ref{sizeratio} with $\rho = 2$ to prove \cref{equal-sized}.

\begin{proof}[Proof of \cref{equal-sized}]
\ipr \ is $(1+\epsilon)(1+\alpha)$-consistent by Lemma~\ref{lem:consistency}. 
To prove the robustness, by Lemma~\ref{lem:nondecre-eq} and \ref{bmax-equalsized}, we can apply Lemma~\ref{sizeratio} with $\rho = 2$, $c_1 = \frac{1}{\ell}$ and $c_2 = 1$. 
Then, $\beta \leq \max\{2, \left(\frac{(\ell-1)}{\ell}+\frac{1}{\alpha}\right) + \frac{{\ell}}{\alpha\hat{s}_{j}\OPTC}\} $. Note that in the proof of Lemma~\ref{sizeratio} we have shown that $\bmin \geq (\alpha/{\ell}){\hat{s}}_{j}\OPTC$, so $\frac{{\ell}}{\alpha\hat{s}_{j}\OPTC} \leq \frac{1}{b_{\min}} \leq 1$. 
Therefore, the robustness bound
$\beta \leq \max\{2, \left(\frac{(\ell-1)}{\ell}+\frac{1}{\alpha}\right) + \frac{{\ell}}{\alpha\hat{s}_{j}\OPTC}\} \le 2 + \frac{1}{\alpha}$.
\end{proof}

\subsubsection{Infinitesimal jobs}

We also consider the “continuous” case with infinitesimal jobs, i.e,  there are infinitely many jobs with processing time $p_j = p$ for some extremely small $p > 0$. In this setting, in the partitioning stage, it's possible to divide the load of all jobs into $m$ bags, each of which can have an arbitrary total processing time. Again, to apply Lemma~\ref{sizeratio}, we need the monotonicity property of $b_{\min}^{(i)}$ (implied by Lemma~\ref{lem:nondecre-eq}) and an upper bound of $b_{\max}$ (proved in Lemma~\ref{bmax-infi}). 

\begin{lemma}\label{bmax-infi}
Let $\bagsipr = \{B_1, \ldots, B_m\}$ be the partition of the $n$ jobs returned by \ipb \ with $ \rho = 2 $. Then $\bmax \leq \frac{W}{\ell}$ in the case of infinitesimal jobs. 
\end{lemma}

\begin{proof}
Since the jobs are infinitesimal, the re-balancing process can be done perfectly, and within each group, all bags have the same total processing time. Therefore, each bag in $\Mmax$ has the same total processing time $W/\ell$.
\end{proof}

Using \cref{sizeratio}, we obtain an improved robustness of $(1+ 1/\alpha)$ for the returned partition by \ipb \ with $\rho = 2$ when all the jobs are infinitesimal.

\begin{theorem}
\label{Infinitesimal} If all jobs are infinitesimal, then, for any constant $\epsilon \in (0,1)$ and any $\alpha \in (0,1)$,  \ipb \ with $\rho = 2$  is a $(1+\epsilon)(1+\alpha)$-consistent and $(1+1/\alpha)$-robust partitioning algorithm.
\end{theorem}
\begin{proof}
\ipr \ is $(1+\epsilon)(1+\alpha)$-consistent by Lemma~\ref{lem:consistency}. 
To prove the robustness, by Lemma~\ref{lem:nondecre-eq} and \ref{bmax-infi}, we can apply Lemma~\ref{sizeratio} with $\rho = 2$, $c_1 = \frac{1}{\ell}$ and $c_2 = 0$. Then we derive that the robustness bound $\beta\leq \max\{2, \left(\frac{(\ell-1)}{\ell}+\frac{1}{\alpha}\right) \} \le 1 + \frac{1}{\alpha}.$
\end{proof}

\subsection{New algorithm and lower bounds for the \spc \ case}
\label{sec:specialcase0-1speeds}

In this section, we consider a special case where each machine has either speed 0 (unavailable) or speed 1 (available).  Recall that $\mupp$ denotes the total number of machines, the speeds of which are either 0 or 1. In this case, the predicted speeds vector $\hat{\s}$ is uniquely determined by the number of ones in the vector, i.e. the number of available machines. Therefore, instead of a predicted speeds vector, we let the prediction be a scalar $\mpred$ which denotes the predicted number of available machines. In addition, we let $\mtrue$ denote the actual number of available machines. An algorithm in the \spc \ special case is $c$-consistent if $\max_{\p, \mtrue} alg(\p, \mtrue, \mtrue)/opt(\p, \mtrue) \leq c$ where $alg(\p, \mpred, \mtrue)$ is the makespan of the schedule returned by the algorithm when it is given prediction $\mpred$ in the first stage and actual number of machines available $\mtrue$ in the second stage. An algorithm is $\beta$-robust if it achieves a $\beta$ approximation ratio when the predictions can be arbitrarily wrong, i.e., if $\max_{\p, \mpred, \mtrue} alg(\p, \mpred, \mtrue)/opt(\p, \mtrue) \leq \beta$.

\subsubsection{A new algorithm}

In this section, we first show that there is a $(1+\epsilon)$-consistent and $2(1+\epsilon)$-robust partitioning algorithm in the binary speeds case. Then, we can pay an extra $(1+\epsilon)$ factor for the approximation ratio by using PTAS on the scheduling stage.

We now briefly introduce the idea of the $(1+\epsilon)$-consistent and $2(1+\epsilon)$-robust partitioning algorithm. First, we run PTAS with accuracy parameter $\epsilon$, assuming that $\mpred$ machines are available, and create $\mpred$ subsets of jobs 
$\F = \{F_1, \ldots, F_{\mpred}\}$. Next, we partition each subset into $\floor{\frac{m}{\mpred}}$ or $\ceil{\frac{m}{\mpred}}$ bags by the LPT algorithm, resulting in a total of $m$ bags. The details are presented in Algorithm~\ref{speed_0-1:general}. 
\begin{algorithm}[H]
\begin{algorithmic}[1]
\caption{The $\{ 0,1 \}$-speed SSP algorithm}\label{speed_0-1:general}

\INP the predicted number of available machines $\mpred$, job processing times $p_1, \ldots, p_n $, accuracy $\epsilon \in (0,1)$

\State Partition jobs with PTAS$(\epsilon)$ into $\mpred$ subsets $\{F_1, \ldots, F_{\mpred}\}$ such that $p(F_1) \geq \cdots \geq p(F_{\mpred})$\label{li:alg3-initial} 
\State \textbf{for} $i = 1$ to $\mpred$ \textbf{do}
\State \quad Partition $F_i$ with LPT into $\floor{\frac{m}{\mpred}}$ or $\ceil{\frac{m}{\mpred}}$ bags such that there are a total of $m$ bags
\State Return $\{B_1, \ldots, B_m\}$
\end{algorithmic}
\end{algorithm}

Recall that $opt(\p, \mtrue)$ and $alg(\p, \mpred, \mtrue)$ are the optimal makespan and the makespan that is obtained by \cref{speed_0-1:general} with job processing times $\p$, $\mpred$ available machines in predictions and $\mtrue$ machines in reality. Note that $opt(\p, \mtrue)$ is just the optimal makespan of the makespan minimization problem on $\mtrue$ identical parallel machines with job processing times $\p$. 
We let $\B^{\ast} = \{B_1^{\ast},B_2^{\ast},\ldots, B_{\mupp}^{\ast}\}$ be the optimal partition of the jobs when all $\mupp$ machines are available in reality, the optimal makespan of which is exactly $opt(\p, m)$.

We first list a few important properties about the optimal makespan in the makespan minimization problems for identical machines. These properties (Lemma~\ref{opt1}, \ref{opt2}  and \ref{machine-impact}) will be used later. 

\begin{lemma}\label{opt1}\cite{SZ18}
    $opt(\p, x) \ge opt(\p, y)$, $\forall x \le y$.
\end{lemma}

\begin{lemma}\label{opt2}\cite{pinedo2012scheduling}
    $opt(\p, x) \ge \max\{\sum_j{p_j} / x, \max_j p_j\} $.
\end{lemma}

\begin{lemma}\cite{rustogi2013parallel}\label{machine-impact}
Consider the makespan minimization problem on identical parallel machines with job processing times $\p$. Given integers $y \geq x$, let $y = ux+v$, where $u,v$ are integers such that $u \geq 1, 1 \leq v \leq x$. Then,
$$\frac{opt(\p, x)}{opt(\p, y)} \leq \ceil{\frac{y}{x}} = u+1.$$
\end{lemma}

In the remainder of this section, we analyze the robustness of \cref{speed_0-1:general} as the consistency is $(1+\epsilon)$ by construction. 
First, we show that the ratio $\frac{\max_{B \in \B}p(B)}{\max_{B'\in B^{\ast}}p(B')} $ of the maximum total processing time of a bag in $\B$ to the maximum total processing time of a bag in $\B^*$ is an upper bound for the robustness of any partition $\B$. 

\begin{lemma}
\label{0-1:robustness}
    Let $\B = \{B_1,B_2,\ldots, B_{\mupp}\}$ be a partition of $n$ jobs with processing times $\p$ into $m$ bags. If $\frac{\max_{B \in \B}p(B)}{\max_{B'\in B^{\ast}}p(B')} \le \tta $ where $\B^*$ is the optimal partition of the jobs when all $m$ machines are available, then $\B$ is a $\max\{\tta,2\}$-robust partition.
\end{lemma}

\begin{proof}
    If the number of nonempty bags in $\B$ is smaller than or equal to $\mtrue$, then by \cref{opt1}, we have $alg(\p, \mpred,\mtrue) = \max_{B\in \B}{p(B)} \le  \tta \max_{B'\in B^{\ast}}p(B') = \tta \cdot opt(\p, m) \leq \tta \cdot opt(\p, \mtrue)$. 
    
    Otherwise, some bags need to be scheduled on the same machine in order to fit $\mtrue$ machines. If we can ``merge" some bags and end up with $\mtrue$ bags after the merging process, then we can simply follow a one-on-one schedule. 
    To achieve this, after learning the number of machines available $\mtrue$, we merge the smallest two non-empty bags until there are $\mtrue$ non-empty bags. Through this process, if the total processing time of the largest bag does not change, then by Lemma~\ref{opt1}, we have $alg(\p, \mpred,\mtrue) = \max_{B\in \B}{p(B)} \le \tta \cdot \max_{B'\in B^{\ast}}{p(B')} = \tta \cdot opt(\p, m) \le \tta \cdot opt(\p, \mtrue)$. 
    Otherwise, if the total processing times of the largest bag increases through merging, then we claim that after the merging process, the total processing time of each bag is at least $b_{1}/2$ where $b_{1}$ is the total processing time of the largest bag. 
    We prove the above claim by contradiction. Assume there exists a bag with total processing time less than $b_{1}/2$ after the merging process. Since we know $b_{1}$ must be the sum of the total processing time of two smallest bags, every other bag must have a load larger than or equal to $b_1/2$. Therefore, we have a contradiction. With the above claim, we have $alg(\p, \mpred,\mtrue) \le \frac{2}{\mtrue+1}\sum_{B\in \B}p(B)$.
    Then by Lemma~\ref{opt2}, if $\sum _j p_j > \mtrue \max_j p_j$, we have $alg(\p, \mpred,\mtrue) \le \frac{2}{\mtrue+1}\sum_{B\in \B}p(B)  \le \frac{2\mtrue}{\mtrue+1} opt(\p, \mtrue) \leq 2 \cdot opt(\p, \mtrue)$, where the second inequality holds because $opt(\p, \mtrue) \geq \sum_j p_j / \mtrue$. 
    Otherwise, if $\sum _j p_j \leq \mtrue \max_j p_j$, then $alg(\p, \mpred,\mtrue) \le \frac{2}{\mtrue+1} \sum_{B\in \B}p(B) \leq \frac{2\mtrue}{\mtrue+1} \max _j p_j \leq 2 \max_j p_j \leq 2 \cdot opt(\p, \mtrue)$, where the second inequality holds because $opt(\p, \mtrue) \geq \max_j p_j$. 
\end{proof}

We now prove a very important property of the LPT algorithm, which we will use along with \cref{0-1:robustness} to prove the robustness of \cref{speed_0-1:general}.

\begin{lemma}
\label{lem-lptratio}
For any job processing times $\p$ and the total number of machines $m$, let $\B_{\textsc{LPT}} = \{B_1, \cdots, B_m\}$ be the partition of $n$ jobs into $m$ bags by the LPT algorithm with $p(B_1) \ge \ldots \ge p(B_m)$. If $|B_1| > 1$, then $p(B_1) \le 2\sum_{B\in \B}p(B) / (m+1) $.  
\end{lemma}

\begin{proof}
Assume, for the sake of contradiction, that $p(B_1) > 2\sum_{B\in \B}p(B) / (m+1)$. We must have $p(B_m) < \sum_{B\in \B}p(B) / (m+1)$; otherwise, $\sum_{B\in \B}p(B) \ge p(B_1) + (m-1) p(B_m) > \sum_{B\in \B}p(B)$ which leads to a contradiction. Let $|B_1| = k > 1$. 
If we sort the jobs in $B_1$ in a non-decreasing order, then the largest $k-1$ jobs in $B_1$ have total processing time at least $\frac{k-1}{k}p(B_1) > (\frac{k-1}{k})2\sum_{B\in \B}p(B) / (m+1) > \sum_{B\in \B}p(B)/(m+1) > p(B_m)$. Since LPT in each iteration always adds a job to the bag with the smallest total processing time, there's a contradiction.
\end{proof}

Using the above lemmas, we can prove that \cref{speed_0-1:general} is a $(1+\epsilon)$-consistent and $2(1+\epsilon)$-robust partitioning algorithm. 

\begin{theorem}\label{ub:1-consist}
For any constant $\epsilon > 0$,
Algorithm~\ref{speed_0-1:general} is a $(1+\epsilon)$-consistent and $2(1+\epsilon)$-robust partitioning algorithm. 
\end{theorem}

\begin{proof}
If $\mtrue = \mpred$, since Algorithm~\ref{speed_0-1:general} would partition all of the jobs into $\mpred$ subsets according to the $\ptas$, it is $(1+\epsilon)$-consistent.
If $\mtrue \neq \mpred$, we compare the returned partition $\B = \{B_1, \ldots, B_{m} \}$ by \cref{speed_0-1:general} with the optimal partition $\B^{\ast} = \{B_1^{\ast}, \ldots, B_m^{\ast}\}$ when all $m$ machines are available. 
We claim that $\max_{B\in \B} p(B) \le 2(1+\epsilon) \cdot \max _{B'\in \B^{\ast}} p(B')$. If this claim is satisfied then by \cref{0-1:robustness}, the robustness is at most $2(1+\epsilon)$.
Let $B_1 = \argmax_{B \in \B} p(B)$, and $F_1$ be the subset that contains $B_1$.
If $|B_1| = 1$, then the bag only has a large job and we have $p(B_1) = \max _{B'\in \B^{\ast}} p(B')$. By \cref{0-1:robustness} the robustness is at most 2. 
Otherwise, if $|B_1| > 1$, since for each subset we partition it into at least $\floor{m/\mpred}$ bags, by Lemma~\ref{lem-lptratio} we have $p(B_1) \le 2 \sum_{B\in F_1}p(B) / (\floor{m/\mpred}+1)$.
By Lemma~\ref{machine-impact}, we have $opt(\p, \mpred) \le \ceil{m/\mpred} opt(\p, m)$.
Then we have 
\begin{align*}
    p(B_1) & \le 2 \sum_{B\in F_1}p(B) / (\floor{m/\mpred}+1) \\
    & \le 2(1+\epsilon)opt(\p, \mpred)/ (\floor{m/\mpred}+1) \\
    &  \le 2(1+\epsilon)\ceil{m/\mpred} opt(\p, m)/(\floor{m/\mpred}+1) \\
    & \le 2(1+\epsilon)opt(\p, m).
\end{align*}

Finally, by \cref{0-1:robustness}, we conclude that Algorithm~\ref{speed_0-1:general} is a $2(1+\epsilon)$-robust partitioning algorithm.
\end{proof}

\subsubsection{Lower bounds}
In this section, we first show that there exists an instance on which any $1$-consistent algorithm in the $\{0,1\}$-speed case has a robustness that is at least $\frac{2(m-1)}{m}$. We then show that there is a consistency-robustness trade-off for any deterministic algorithm in the $\{0,1\}$-speed case.

\begin{lemma}
Any 1-consistent algorithm for the $\{0,1\}$-speed SSP problem has a robustness at least $\frac{2(m-1)}{m}$, where $m$ is the total number of machines.
\end{lemma}

\begin{proof}
Consider an instance with a total of $m$ machines, $\mpred = m$ and $n = m(m-1)$ unit-sized jobs, i.e. $p_j = 1$ for all $j \in \{1, \ldots, n\}$.
    
We first argue that a $1$-consistent algorithm must partition the jobs into $m$ bags with each bag having exactly $m-1$ job.
Assume the predictions are correct, i.e., $\mtrue = m$.  By Lemma~\ref{opt2}, the optimal makespan is at least $m-1$. Therefore, such a partition can obtain the optimal makespan $m-1$. Now, consider a partition $B_1, \ldots, B_m$ of the jobs such that some bag does not contain exactly $m-1$ jobs. In this case, there exists at least one bag that contains more than $m-1$ jobs. Since the optimal makespan is $m-1$, an algorithm that makes such a partition is not $1$-consistent. 

Next, consider the case where the predictions are incorrect and the actual number of available machines is $\mtrue = m-1$. Then the 1-consistent algorithm has to schedule two bags, each of which contains $m-1$ jobs, on the same machine. Therefore, the makespan of the 1-consistent algorithm is $2(m-1)$ if $\mtrue = m-1$. However, the optimal schedule if $\mtrue = m-1$ is to schedule $m$ jobs on each machine. Therefore the robustness is $2(m-1) / m$.    
\end{proof}

\begin{theorem}\label{lb:0-1_speeds}
For any $\alpha \in [0,1/2)$, if a deterministic algorithm for the $\{0,1\}$-speed SSP problem is $(1+\alpha)$-consistent, then its robustness is at least $(4-2\alpha)/{3}$.
\end{theorem}
\begin{proof}
Consider an instance with $\mpred = m = 3$ and $n = 6k$ unit-sized jobs, i.e. $p_j = 1$ for all $j \in \{1, \ldots, n\}$ where $k$ is a large positive integer. 
Let $\B = \{B_1,B_2,B_3\}$ be the partition generated by an arbitrary $(1+\alpha)$-consistent algorithm with $p(B_1) \ge p(B_2) \ge p(B_3)$. We let $alg(\p, 3, \mtrue)$ denote the makespan achived by the $(1+\alpha)$-consistent algorithm when there are $\mtrue$ machines available in reality, with prediction being $\mpred = 3$.
We first prove by contradiction that we must have $alg(\p, 3, 3) = p(B_1)$ when $\mtrue = 3$. Assume that $alg(\p, 3, 3) = p(B_i) + p(B_j)$ where $i\neq j, i,j \in \{1,2,3\}$. Since the algorithm of our interest is $(1+\alpha)$-consistent, we have $alg(\p, 3, 3) \le (1+\alpha)opt(\p, 3) = (1+\alpha)2k < 3k$. Then for the other bag $B_t$, the total processing time of it is $p(B_t) > n - (p(B_i) + p(B_j)) = 3k$, which contradicts that $alg(\p, 3, 3) = p(B_i) + p(B_j) < 3k$. 

Now we consider the approximation ratio when $\mtrue = 2$. In this case, the algorithm would schedule $B_2$ and $B_3$ on the same machine. Then the makespan of the algorithm is $alg(\p, 3, 2) = p(B_2) + p(B_3) \ge 6k - (1+\alpha)2k = 4k-2\alpha k$ with $opt(\p, 2) = 3k$. Thus, the approximation ratio is at least $\frac{alg(\p, 3, 2)}{opt(\p, 2)} \geq \frac{4-2\alpha}{3}$. 
\end{proof}

\section{Experiments}
\label{sec:experiments}

We empirically evaluate the performance of \ipb \ on synthetic data against benchmarks that achieve either the best-known consistency or the best-known robustness for SSP.

\subsection{Experiment settings}

\paragraph{Benchmarks.} We compare three  algorithms.  
\begin{itemize}
\item \ipb \ is Algorithm~\ref{alg-general} with $\rho = 4$ and  $\alpha = 0.5$.
\item The Largest Processing Time first partitioning algorithm, which we call \textsc{LPT-Partition}, creates $m$ bags by adding each job, in decreasing order of their processing time, to the bag with minimum total processing time. \textsc{LPT-Partition} is $2$-robust (and $2$-consistent since it ignores the predictions) \cite{EHMNSS20}.
\item The \textsc{1-consistent} algorithm  completely trusts the prediction and generates a partition that is $1$-consistent (but has arbitrarily poor robustness due to our lower bound in Proposition~\ref{1-consist}). 
\end{itemize}

In practice, PTAS algorithms for scheduling are extremely slow. Instead of using a PTAS for the scheduling stage, we give an advantage to the two benchmarks  by solving their scheduling stage  via integer programming (IP). However, since we want to ensure that our algorithm has a polynomial running time, we use the LPT algorithm to compute a schedule during both  the partitioning and scheduling stage of \ipb, instead of a PTAS or an IP. To compute the approximation ratio  achieved by the different algorithms, we compute the optimal solution using an IP.

\paragraph{Data sets.} In the first set of experiments, we generate synthetic datasets with $n = 50$ jobs and $m = 10$ machines and evaluate the performance of the different algorithms as a function of the standard deviation of the  the prediction error distribution. The job processing times $p_j$ are generated i.i.d. either from $\U(0, 100)$, the uniform distribution in the interval $(0, 100)$, or $\N(50, 5)$, the normal distribution  with mean $\mu_p = 50$ and standard deviation $\sigma_p = 5$.  The machine speeds $s_i$ are also generated i.i.d., either from $\U(0, 40)$ or $\N(20, 4)$. We evaluate the performance of the algorithms over each of the $4$ possible combinations of job processing time and machine speed distributions. The prediction error  $err(i) = \hat{s}_i - {s}_i$ of each machine is sampled i.i.d. from $\N(0, x)$ and we vary $x$ from $x = 0$ to $x = \mu_s$ (the mean of machine speeds).

In the second set of experiments, we fix the distributions of the processing times, machine speeds, and prediction errors to be $\N(50, \sigma_p)$, $\N(20, \sigma_s),$ and $\N(0, 4)$ respectively, with default values of $\sigma_p = 5$ and $\sigma_s = 4$. We evaluate the algorithms' performance as a function of (1) the number $n$ of jobs, (2) the number $m$ of machines, (3) $\sigma_p$, and (4) $\sigma_s$. For each figure, the approximation ratio achieved by the different algorithms are averaged over $100$ instances generated i.i.d. as described above. Additional details of the experiment setup are provided in Appendix~\ref{sec:appexperiments}.

 \begin{figure*}[t!]
 	\centering
    \subfigure{\includegraphics[width=0.24\textwidth]{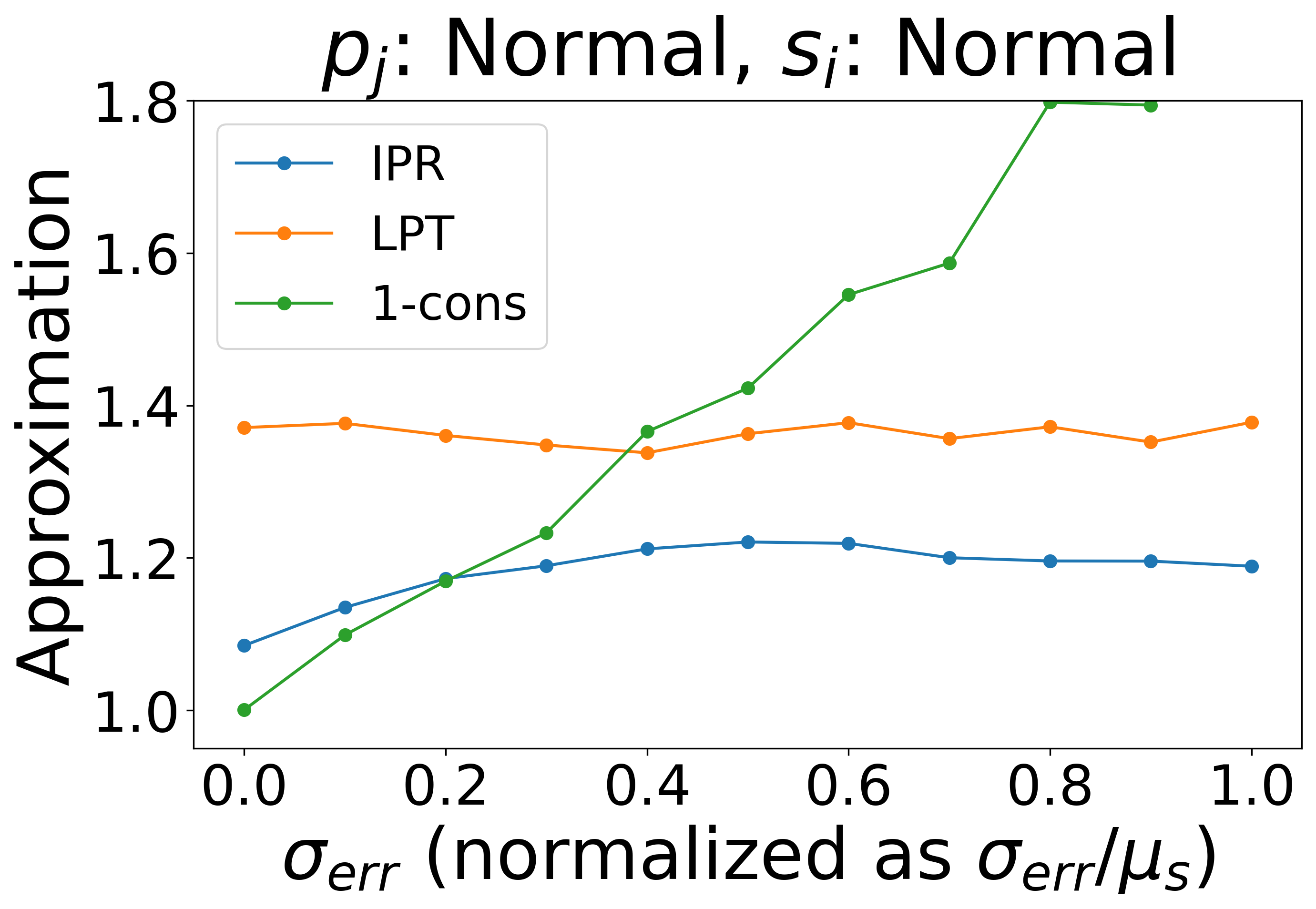}} 
    \subfigure{\includegraphics[width=0.24\textwidth]{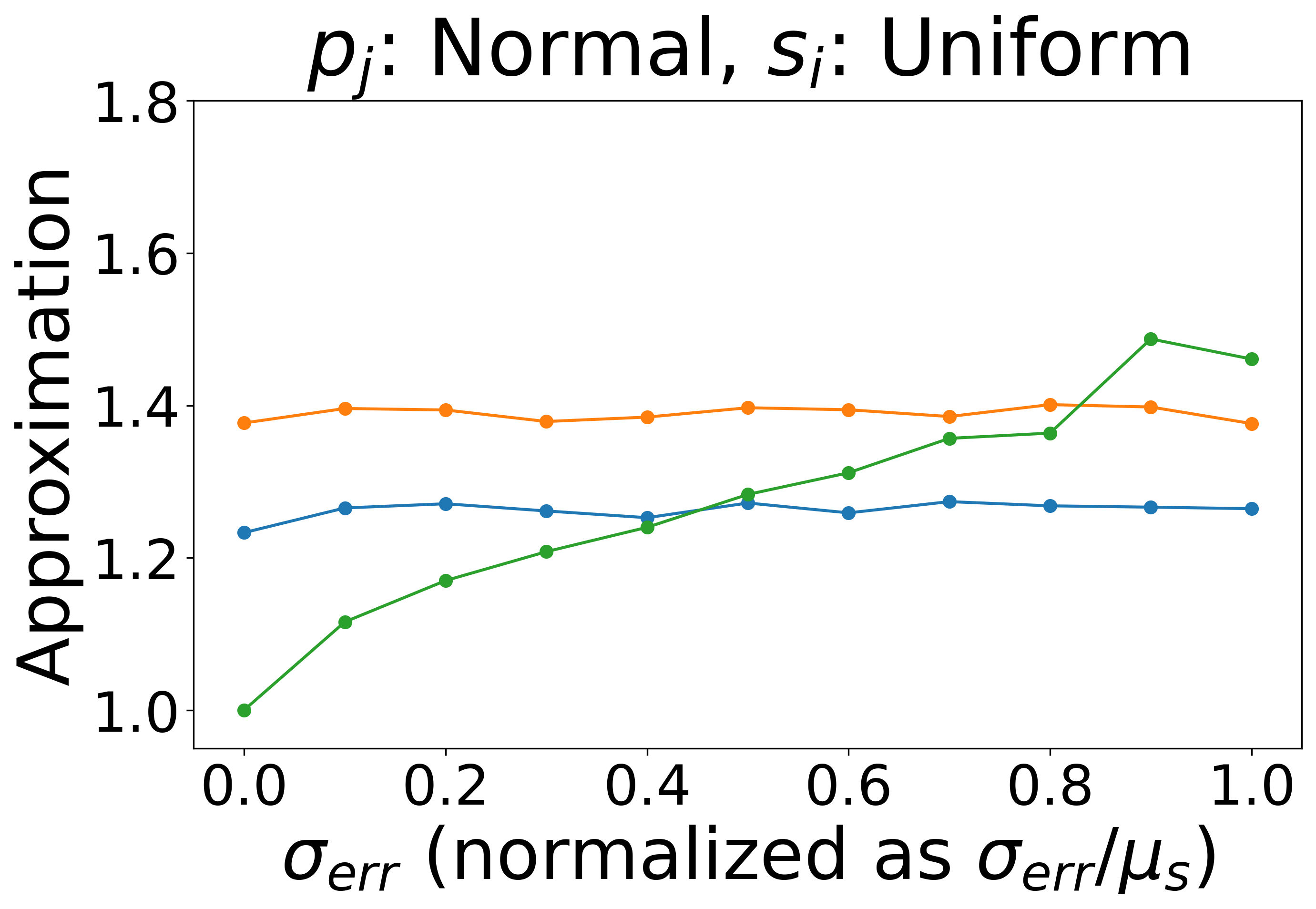}} 
    \subfigure{\includegraphics[width=0.24\textwidth]{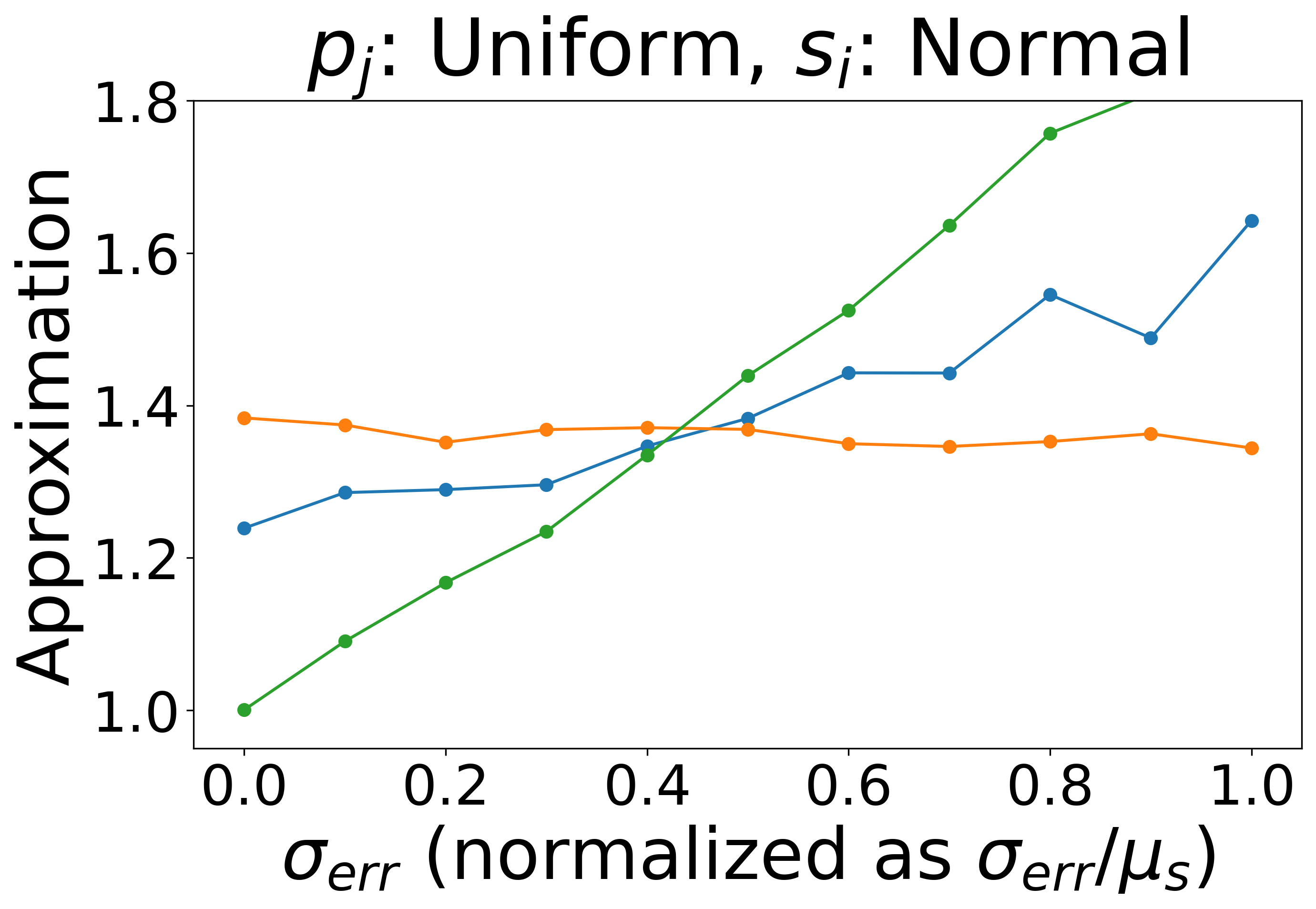}}
    \subfigure{\includegraphics[width=0.24\textwidth]{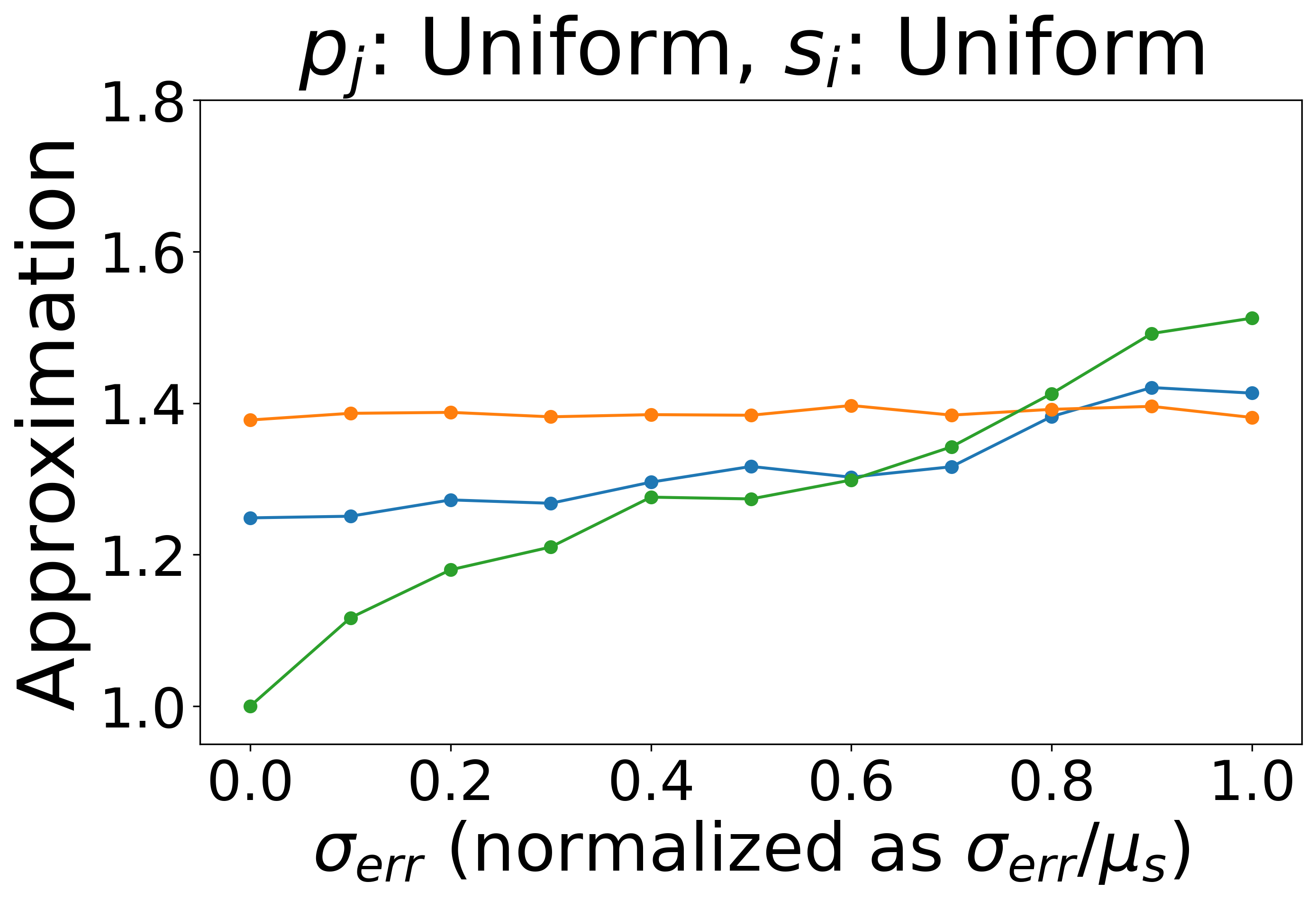}}
    \subfigure{\includegraphics[width=0.24\textwidth]{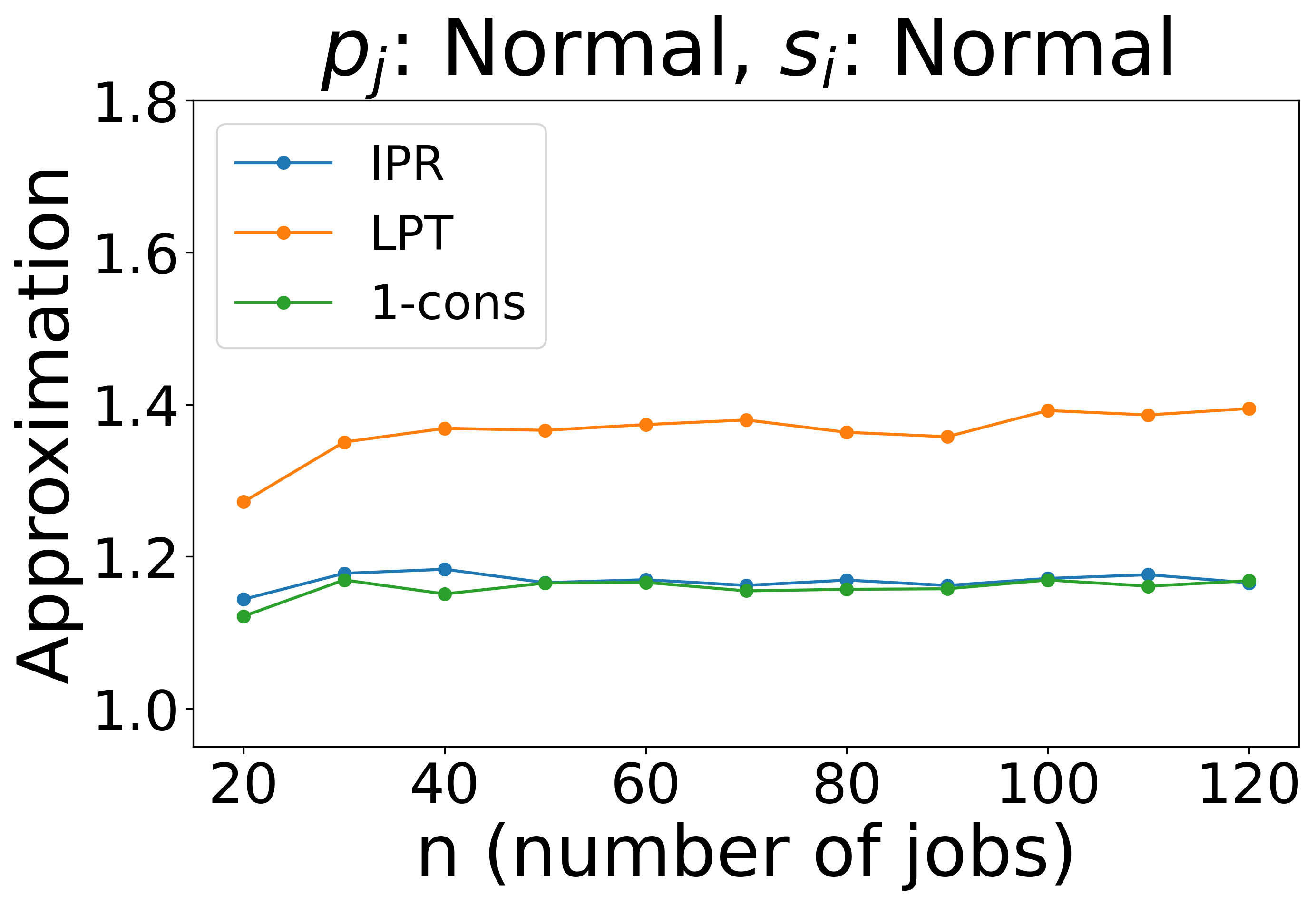}}
    \subfigure{\includegraphics[width=0.24\textwidth]{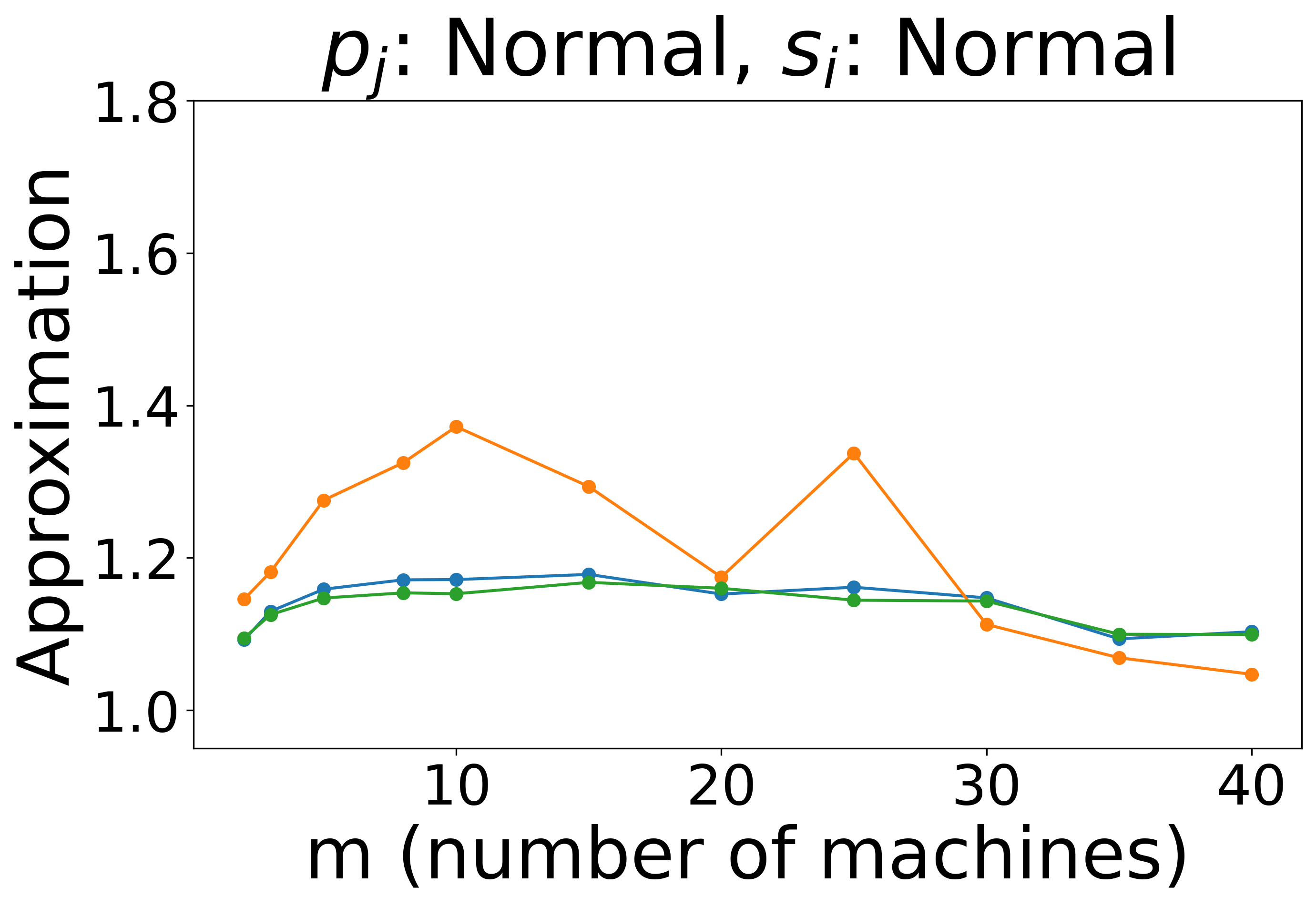}}
    \subfigure{\includegraphics[width=0.24\textwidth]{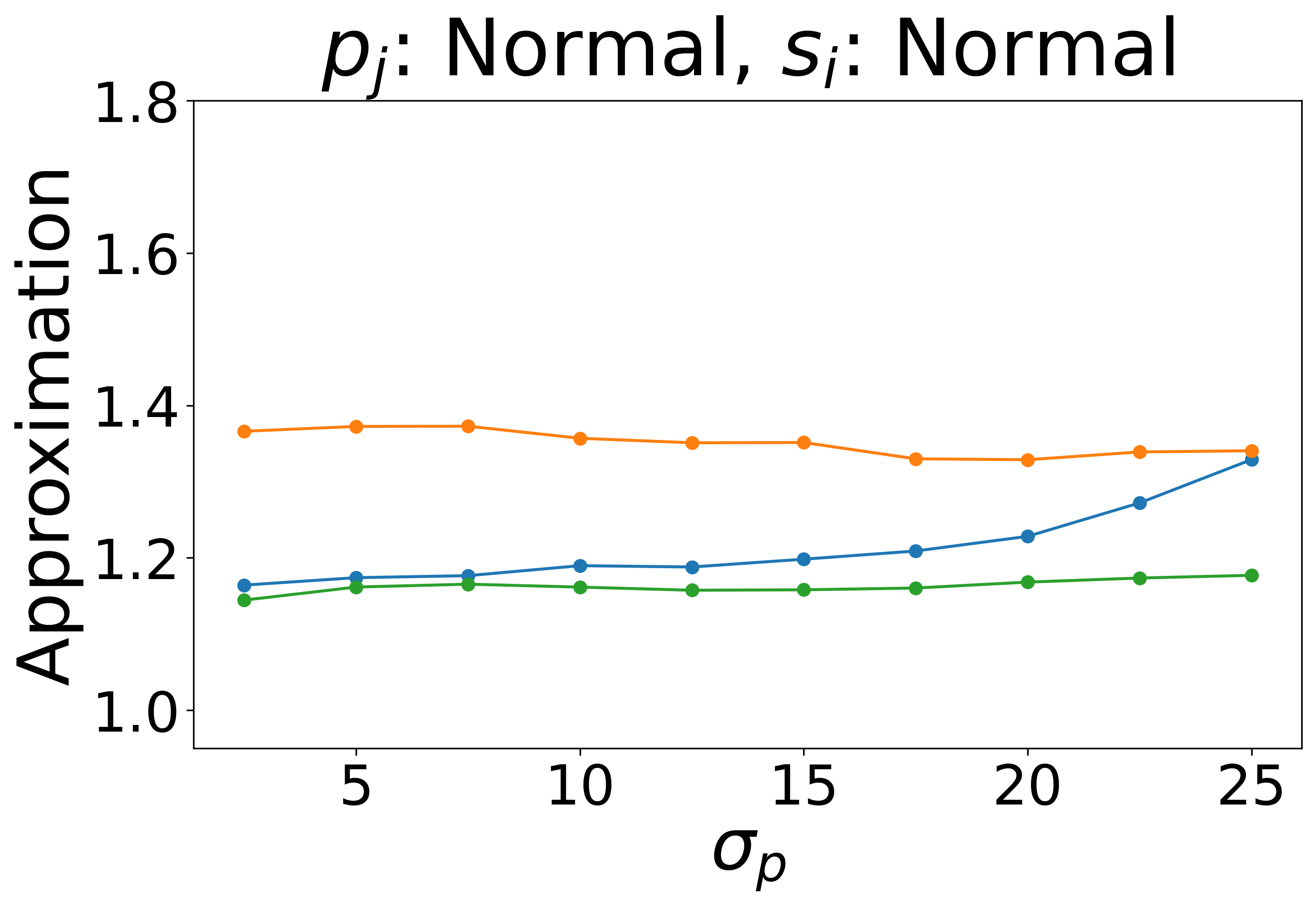}} 
    \subfigure{\includegraphics[width=0.24\textwidth]{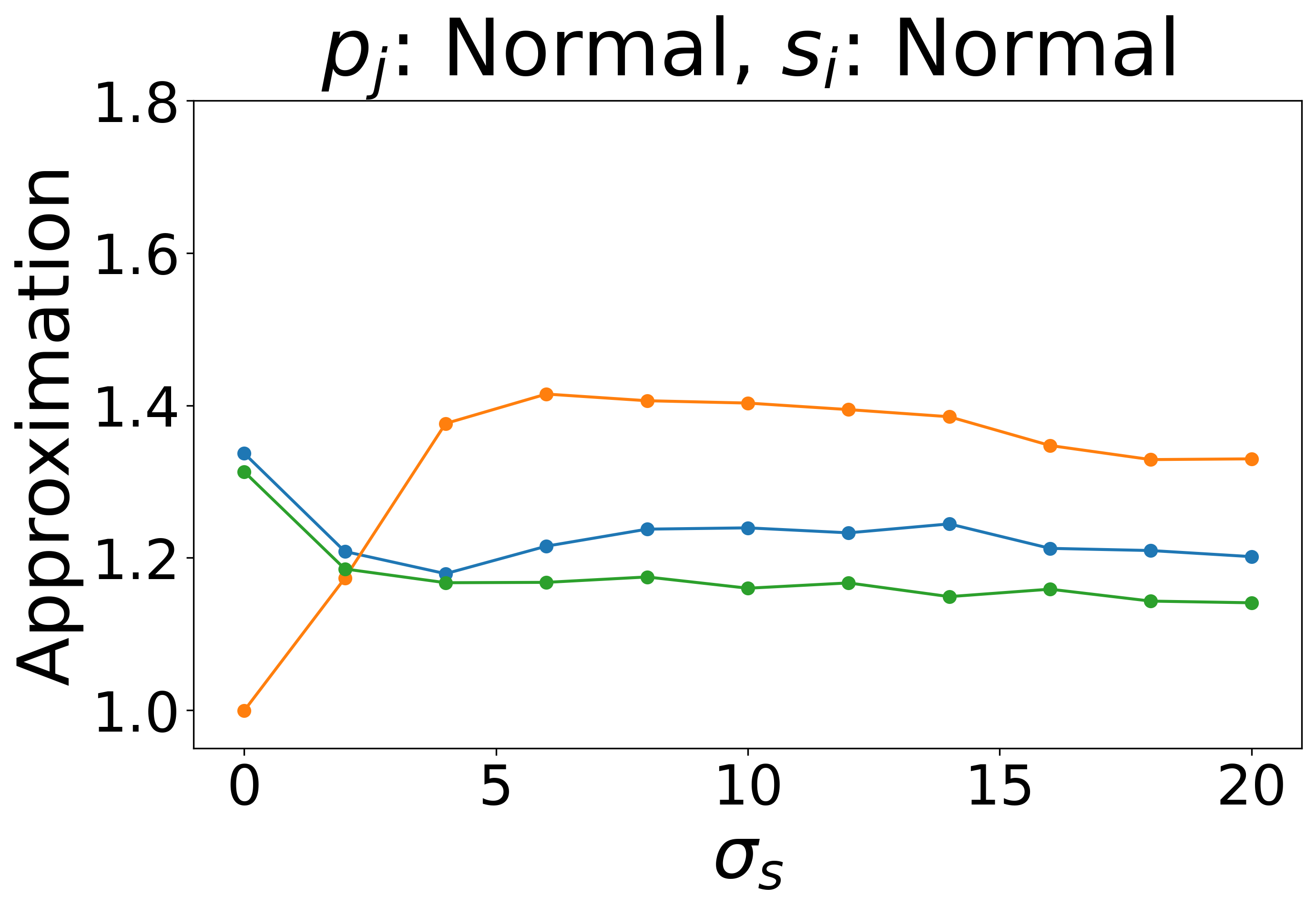}}
 	\caption{The approximation ratio achieved by  our algorithm, \ipb, and the two benchmarks as a function of   the standard deviation of the prediction error $\sigma_{err}$  for different job processing time and true speed distributions (row 1) and as a function of the   number of jobs $n$, the   number of machines $m$, the standard deviation $\sigma_p$ of the job processing time distribution, and the standard deviation $\sigma_s$ of the true speed distribution (row 2).}
 	\label{fig:exp}
 \end{figure*}

\subsection{Experiment results}

\paragraph{Experiment set 1.}
From the first row of Figure~\ref{fig:exp}, we observe that, in all four settings,  when we vary the magnitude of the prediction error, \ipb \ outperforms  \lptpar \ when the error is small and outperforms  \textsc{1-consistent}  when the error is large. Since \textsc{LPT-Partition} does not use the predictions, its performance remains constant as a function of the prediction errors. Since \textsc{1-consistent} completely trusts the predictions, it is optimal when the predictions are exactly correct but its performance deteriorates quickly as the prediction errors increase.

\ipb \ combines the advantages of \textsc{LPT-Partition} and \textsc{1-consistent}: when the predictions are relatively accurate, it is able to take advantage of the predictions and outperform \lptpar. When the predictions are increasingly inaccurate, \ipb \ has a  slower deterioration rate compared to \textsc{1-consistent}. It is noteworthy that, in some settings, \ipb \  simultaneously outperforms both benchmarks for a wide range of values of the standard deviation $\sigma_{err}$ of the prediction error distribution. When the distributions of job processing times and machine speeds are  $\N(50, 5)$ and $\N(20, 4)$ respectively, \ipb \ achieves the best performance  when $\sigma_{err} / \mu_s \geq 0.2$. When they are $\N(50, 5)$ and $\U(0, 40)$, \ipb \ outperforms both benchmarks when $\sigma_{err} / \mu_s \geq 0.4$.

 \paragraph{Experiment set 2.}
The number of 
jobs has almost no impact on the performance of any of the algorithms. 
However, the approximations achieved by the algorithms do improve as the number of machines $m$ increases, especially for $\lptpar$. The reason is that $m$ is also the number of bags, so when the number of bags increases, there is more flexibility in the scheduling stage, especially when the total processing times of the bags are balanced.

 $\ipb$ is the algorithm most sensitive to the standard deviation $\sigma_p$ of the job processing times. It has performance close to that of $\onecons$ when $\sigma_p$ is small, and similar to  $\lptpar$ when $\sigma_p$ is large. 
 The  approximation ratio of  \lptpar \ increases as  $\sigma_s$ increases, while our algorithm and the \textsc{1-consistent} partitioning algorithm are relatively insensitive to the change in $\sigma_s$. Since the \lptpar \ algorithm generates balanced bags of similar total processing times, it performs well when the machine speeds are all almost equal, but its performance then   quickly degrades  as $\sigma_s$   increases.

\paragraph{The parameter $\alpha$.}
We  also study  the impact of the parameter $\alpha$, which controls how much \ipb \ trusts the predictions, in the setting that is identical to experiment set 1.


  \begin{figure}[h!]
 	\centering
    \subfigure{\includegraphics[width=0.3\textwidth]{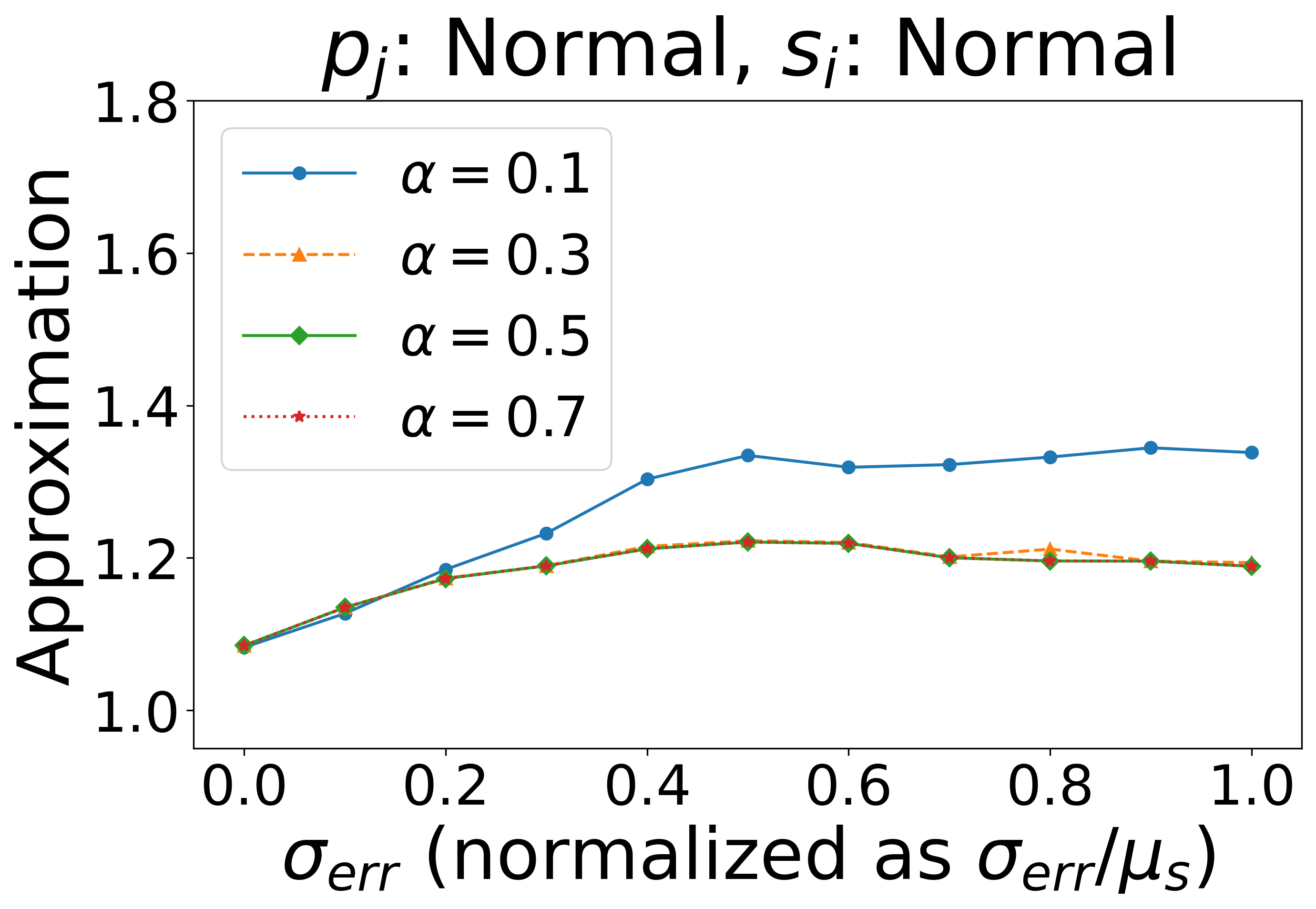}} 
 	\caption{Approximation ratio of \ipb \ for different values of the parameter $\alpha$ as a function of the standard deviation $\sigma_{err}$ of the distribution for the prediction errors. The lines for $\alpha \in \{0.3, 0.5, 0.7\}$ almost completely overlap.}
 	\label{fig:varyalpha}
 \end{figure}
 
 From \cref{fig:varyalpha}, we observe that when the prediction error is small, \ipb \ performs  well for all values of $\alpha$. In particular, the approximation ratio does not degrade as $\alpha$ increases, which, when the prediction error is small, is unlike our theoretical results.  However, when the prediction errors are large, \ipb \ performs worst with a small value of $\alpha = 0.1$. For other values of $\alpha$ such that $\alpha \geq 0.3$, the approximation ratios are almost identical. Thus, empirically, we do not observe a significant trade-off between the robustness and consistency achieved by \ipb \ as a function of $\alpha$ and its performance is identical for $\alpha = 0.3$, $\alpha = 0.5$, and $\alpha = 0.7$.

%
%
%
 \bibliographystyle{splncs04}
%
\newpage

\bibliography{references}

\newpage
\appendix
\onecolumn


\section{Additional Discussion of the Model}

\subsection{Learning the speeds of the machines}
\label{sec:applearning}

The problem of learning the machine speeds can be formulated as follows. Each instance of our scheduling problem is associated with a set of features (e.g., date, time, location, \ldots) and we let $F$ denote the set of all possible features. Given historical data of past instances with features $f \in F$ and machine speeds $\textbf{s} \in \mathbb{R}^m$, the learning problem consists of finding a function $h : F \rightarrow \mathbb{R}^m$ that maps features to machine speeds. This function $h$ can be learned using, e.g., linear regression or a neural network and how well $h$ is learned depends on many factors, including the amount of historical data, the number of features, and the hypothesis class for $h$.

\subsection{Choice of the prediction error}
\label{sec:apperror}

We define the prediction error to be $\eta = \max_{ i \in [m]} \frac{\max\{\hat{s}_i, s_i\}}{\min\{\hat{s}_i, s_i\}}$, i.e., the $\ell_\infty$ norm of the ratio of the predicted speed and true speed of each machine.
Other prediction errors (such as $\ell_1$) are also reasonable to consider. However, there are unavoidable tradeoffs between different errors, e.g., the $\ell_1$ norm does not distinguish between instances with  one large prediction error and many small prediction errors while $\ell_\infty$  does not distinguish between one large prediction error and many large prediction errors.  The reason behind the choice of the $\ell_\infty$ norm is that the hard instances are those with one (or more) machine speed with a large prediction error (as illustrated by the lower bounds), which can be captured with the $\ell_\infty$ norm. If there are many machine speeds with a small error (in which case $\ell_\infty$ is small and $\ell_1$ is large) then the algorithm performs very well.  However, we believe that studying the scheduling with speed predictions problem that we introduce under other error measurements, such as $\ell_1$, is an interesting direction for future work.

\section{Missing Analysis from Section~\ref{sec:lower-bound}}
\label{applowerbound}


\textbf{\cref{1-consist}.} \textit{For any $n > m$, there is no algorithm that is $1$-consistent and $\frac{n-m+1}{\lceil n/m\rceil}$-robust, even in the case of  equal-size jobs. In particular, for $m = n / 2$, there is no algorithm that is $1$-consistent and $o(n)$-robust.}

\begin{proof}
    Consider the instance with $m$ machines which have predicted speeds $\hat{s}_1 = n - m + 1$ and $\hat{s}_i = 1$ for $i \in \{2, \ldots, m\}$ and $n$ jobs which have processing time $p_j = 1$ for $j \in \{1, \ldots, n\}$.
    
    We first argue that a $1$-consistent algorithm must partition the jobs into  $m-1$ bags that each contain one job and one bag that contains $n - m + 1$ jobs. Assume the predictions are correct, i.e., $s_1 = n - m + 1$ and $s_i = 1$ for $i \in \{2, \ldots, m\}$. Consider the schedule of the bags of such a partition where the machine with speed $s_1 = n - m + 1$ receives the bag that contains $n - m + 1$ jobs and each other machine receives one bag that contains one job. The makespan of this schedule is $1$. Now, consider a partition $B_1, \ldots, B_m$ of the jobs into bags such that there are not $m-1$ bags, each of which contains one job. In this case, since $n - m + 1 > 1$, there are at most $m - 2$ bags that contain one job. We get the following lower bound on the sum of the sizes of the bags of size at least two: $\sum_{i : |B_i| \geq 2} |B_i| = n - \sum_{i : |B_i| = 1} |B_i| \geq n - m + 2.$    Consider a schedule of these bags on the machines. There are two cases, depending on whether the machine with speed $s_1 = n - m + 1$ receives all the bags of size at least $2$. If it does, then the completion time for this machine is at least $(\sum_{i : |B_i| \geq 2} |B_i|) / s_1 > 1$. Otherwise, there is another machine with speed $s_i = 1$ that receives at least one bag of size at least two. In that case, the completion time for this machine is at least $2$. Thus, the makespan of any schedule for such a partition is strictly greater than $1$. Since the optimal makespan is $1$, an algorithm that makes such a partition is not $1$-consistent. 
    
    Next, consider the case where the predictions are incorrect and the speeds are $s_i =1$ for all $i \in \{1, \ldots, m\}.$ In addition, consider a partition $B_1, \ldots, B_m$ of the jobs into $m$ bags such that  $|B_i| = \lfloor n/m\rfloor$ or $|B_i| = \lceil n/m\rceil$ for all $i \in \{1, \ldots, m\}$ and a schedule such that each machine receives one bag. The makespan of such a schedule is $\lceil n/m\rceil$. However, for a partition where there is one bag that contains $n-m+1$ jobs, any schedule of the bags on the machines has makespan that is at least $n-m+1$.
    
    We conclude that any $1$-consistent algorithm has robustness that is at least $\frac{n-m+1}{\lceil n/m\rceil}$.
\end{proof}



\noindent \textbf{\cref{LB-2-speeds}.} \textit{For any $\alpha \in (0, 1)$, if a deterministic algorithm for \ssp \ is $(1+\alpha)$-consistent,  then its robustness is at least $1 + \frac{1-\alpha}{2\alpha} - O(\frac{1}{m})$
    , even in the case where the jobs have equal processing times. In the special case where the processing times are infinitesimal, the robustness of a deterministic $(1+\alpha)$-consistent algorithm is at least $1 + \frac{(1-\alpha)^2}{4\alpha}-  O(\frac{1}{m}).$ }
\begin{proof}
    We first prove the lower bound for the case where all jobs have equal processing times.
    
    Consider the instance of $m$ machines with predicted speeds $\hat{s}_1 = m$ and $\hat{s}_i = 1$ for $i \in \{2, \ldots, m\}$ and $n = 2m-1$ jobs with processing time $p_j = 1$ for $j \in \{1, \ldots, n\}$. We then consider the schedule returned by any $(1+\alpha)$-consistent algorithm, where $\alpha \in (0,1)$. Since $\alpha < 1$, each of the $(m-1)$ machines with predicted speeds $\hat{s}_i = 1$ can be scheduled to process at most one non-empty bag that contains exactly one job. Otherwise, the total completion time on such machine is at least two and the makespan is at least two. Since the optimal makespan is $1$ when $\s = \spred$, the consistency being at least two contradicts the assumption that the algorithm is $(1+\alpha)$-consistent. 
    
    Let $x$ be the number of bags scheduled on the $m-1$ machines with predicted speeds $\hat{s}_i = 1$. As we have argued before, those bags scheduled on the machines with $\hat{s}_i = 1$ cannot contain more than one job. Then the total number of the jobs scheduled on the $m-1$ machines with $\hat{s}_i = 1$ is at most $x$, and the total number of jobs schedule on the machine with speed $\hat{s}_1 = m$ is at least $n-x$. Since the algorithm is $(1+\alpha)$-consistent, we have $n - x = 2m-1-x \leq (1+\alpha) m$, which means $x \geq 2m-1-(1+\alpha) m$.
    
    Because the algorithm partitions at least $(n-x)$ jobs into $(m - x)$ bags, by an averaging argument, the number of jobs in the largest bag is at least: 
    \begin{align*}
        \frac{2m-1-x}{m-x} &= \frac{2(m-x) + x - 1}{m-x}\\
        &= 2 + \frac{x-1}{m-x}\\ 
        &= 1 + \frac{m-1}{m-x}.
    \end{align*}   
    
    By plugging the inequality $x \geq 2m-1-(1+\alpha) m$ into the above formula, we can obtain a lower bound on the largest size of a bag: 
    \begin{align*}
        1 + \frac{m-1}{m-x} &\geq 1 + \frac{m-1}{m-(2m-1-(1+\alpha) m)}\\ 
        &=  1 + \frac{m-1}{\alpha m + 1} \\ &= \frac{(1+\alpha)m}{ \alpha m + 1}.
    \end{align*}
    
    We now consider the situation where the predictions are incorrect and the speeds are $s_i =1$ for all $i \in \{1, \ldots, m\}.$ The optimal partition $B_1, \ldots, B_m$ of the jobs into $m$ bags has size $|B_i| = 2$ for $i \in \{1, \ldots, m-1\}$ and $|B_m| =1$ with optimal makespan $2$. However, the makespan of a $(1+\alpha)$-consistent algorithm is at least $\frac{(1+\alpha)m}{\alpha m + 1}$ because the largest bag has size at least $\frac{(1+\alpha)m}{\alpha m + 1}$. Therefore, the robustness is at least
    \begin{align*}
        \frac{\frac{(1+\alpha)m}{ \alpha m + 1}}{2} & = \frac{(1+\alpha)}{2\alpha + 2/m} \\
        & = \frac{1}{2\alpha+ 2/m} + \frac{1}{2+ 2/(m\alpha)} \\
        & = \frac{1}{2\alpha}\left(\frac{1}{1+1/(m\alpha)}\right) + \frac{1}{2}\left(\frac{1}{1+1/(m\alpha)}\right) \\
        & \geq \frac{1}{2\alpha}\left(1-1/(m\alpha)\right) + \frac{1}{2}\left(1-1/(m\alpha)\right)\\
        & = \frac{1}{2\alpha} + \frac{1}{2} - \left(\frac{1}{2}+\frac{1}{2\alpha}\right)\left(\frac{1}{m\alpha}\right)\\
        & = 1 + \frac{1-\alpha}{2\alpha} - \frac{1+\alpha}{2\alpha^2}\left(\frac{1}{m}\right).
    \end{align*}
    
    We conclude that when jobs have equal processing times, for any $\alpha \in (0,1)$, if an algorithm for SSP is $(1+\alpha)$-consistent, then its robustness is at least $1 + \frac{1-\alpha}{2\alpha} - O(\frac1m)$.
    
    
Next, we prove the lower bound when the jobs are infinitesimal. Consider the instance of $m$ machines with predicted speeds $\hat{s}_1 = m$ and $\hat{s}_i = 1$ for $i \in \{2, \ldots, m\}$ and the input job set $J$ which has total processing time $\sum_{j\in J}{p_j} = 2m-1$. We then consider the schedule of any $(1+\alpha)$-consistent algorithm, where $\alpha \in (0,1)$. Since $\alpha < 1$, each machine with predicted speeds $\hat{s}_i = 1$ for $i \in \{2, \ldots, m\}$ can be scheduled to process jobs of total processing time at most $(1+\alpha)$. Otherwise, the completion time of such machine is more than $(1+\alpha)$ and the makespan is more than $(1+\alpha)$ when the predictions are correct. Since the optimal makespan is $1$ when $\s = \spred$, this leads to a consistency larger than $(1+\alpha)$ which contradicts our assumption.

Let $x$ be the number of bags scheduled on the $m-1$ machines with predicted speeds $\hat{s}_i = 1$. Since each machine with $\hat{s}_i = 1$ can be scheduled with jobs of total processing time at most $(1+\alpha)$, each bag scheduled on these machines cannot have total processing time more than $(1+\alpha)$. Thus, the total processing time of the jobs scheduled on all of the $m-1$ machines with $\hat{s}_i = 1$ is at most $x(1+\alpha)$, and the total processing time of the jobs scheduled on the machine with speed $\hat{s}_1 = m$ is at least $2m-1-x(1+\alpha)$. Since the algorithm is $(1+\alpha)$-consistent, we have $2m-1-x(1+\alpha) \leq (1+\alpha) m$, which means $x \geq \frac{2m-1-(1+\alpha)m}{1+\alpha}$.

Because the algorithm partitions the remaining jobs of total processing time at least $(2m-1-x(1+\alpha))$ into $(m - x)$ bags, by an averaging argument, the maximum total processing time of a bag in the partition is at least $\frac{2m-1-(1+\alpha)x}{m-x}.$

Next, we plug the inequality $x \geq \frac{2m-1-(1+\alpha)m}{1+\alpha}$ into the above formula, and we can obtain 
the lower bound on the maximum total processing time of a bag as follows:
\begin{align*}
    \frac{2m-1-(1+\alpha)x}{m-x} & \geq
    \frac{2m-1-(2m-1-(1+\alpha)m)}{m-
    (2m-1-(1+\alpha)m) / 1+\alpha} \\
    & = \frac{(1+\alpha)m}{m-
    (2m-1-(1+\alpha)m) / 1+\alpha} \\
    & = \frac{(1+\alpha)^2 m}{(1+\alpha)m - (2m-1-(1+\alpha)m)} \\
    & = \frac{(1+\alpha)^2 m}{2\alpha m + 1}.
\end{align*}

We now consider the situation where the predictions are incorrect and the speeds are $s_i =1$ for all $i \in \{1, \ldots, m\}.$ Then the optimal partition $B_1, \ldots, B_m$ of the jobs into $m$ bags has $p(B_i) = \frac{2m-1}{m} = 2-\frac{1}{m}$ for $i \in \{1, \ldots, m\}$ with optimal makespan $2-\frac{1}{m}$. However, the makespan of a $(1+\alpha)$-consistent algorithm is at least $\frac{(1+\alpha)^2 m}{2\alpha m + 1}$ because the maximum total processing time of a bag is at least $\frac{(1+\alpha)^2 m}{2\alpha m + 1}$. Thus, the robustness is at least $\frac{(1+\alpha)^2 m}{2\alpha m + 1}/(2-\frac{1}{m})$ and we have:
\begin{align*}
        \frac{(1+\alpha)^2 m}{(2\alpha m + 1)(2-\frac{1}{m})} & = \frac{(1+\alpha)^2 m}{4\alpha m + 2 - 2\alpha - 1/m}\\
        & = \frac{2\alpha + 1 + \alpha^2}{4\alpha + \frac{2}{m} - \frac{2\alpha}{m} - \frac{1}{m^2}} \\
        & = \left(\frac{1}{2} + \frac{1}{4\alpha} + \frac{\alpha}{4} \right)\left(\frac{1}{1+ \frac{2}{4\alpha m} -\frac{2\alpha}{4\alpha m} - \frac{1}{4 \alpha m^2}}\right)\\
        & \geq \left(\frac{1}{2} + \frac{1}{4\alpha} + \frac{\alpha}{4} \right)\left(1-\left(\frac{1}{2\alpha m}-\frac{1}{2m} - \frac{1}{4 \alpha m^2} \right)\right)\\
        & = \left(\frac{1}{2} + \frac{1}{4\alpha} + \frac{\alpha}{4} \right) - \left(\frac{\alpha^2+2\alpha+1}{4\alpha}\right)\left(\frac{2m-2\alpha m - 1}{4 \alpha m^2}\right)\\
       & = 1 + \frac{(1-\alpha)^2}{4\alpha} - \left(\frac{\alpha^2+2\alpha+1}{4\alpha}\right)\left(\frac{2m-2\alpha m - 1}{4 \alpha m^2}\right).
\end{align*}

We conclude that when jobs are infinitesimal, for any $\alpha\in (0,1)$, if an algorithm for SSP is $(1+\alpha)$-consistent, then its robustness is at least $1 + \frac{(1-\alpha)^2}{4\alpha} - O(\frac{1}{m})$.
\end{proof}


\section{Missing Analysis from Section~\ref{sec:alg-performance}}
\label{sec:appanalysis}

\textbf{\cref{lem:robustnessmainlemma}.} \textit{Let $\B = \{ B_1, \cdots, B_m \}$ be a partition of $n$ jobs with processing times $p_1, \ldots p_n$ into $m$ bags. Then $\B$ is a $\max\{2, \beta(\B) \}$-robust partition, where $\beta(\B) = \frac{\max_{B \in \B, |B| \ge 2}p(B)}{  \min_{B \in \B} p(B)}$.}

\begin{proof}

    In the proof, we use $\mathcal{S}^*$ to denote an optimal schedule of $n$ jobs with processing time $p_1, p_2, \ldots, p_n$, and let $\mathcal{S}$ denote our schedule with partition $\mathcal{B}$. We let $t^*$ and $t$ denote the makespans obtained by $\mathcal{S}^*$ and $\mathcal{S}$, respectively. Finally, we let $\C = \{  \{ j \} \in \B : p_j > \max_{B \in \B, |B| \ge 2} p(B) \}$ be the set of large singleton bags in $\B$.
    
    By applying a scaling argument, we may assume without loss of generality that the machine speeds satisfy $ \sum_{j\in [n]}{p_j} = \sum_{i \in [m]}{s_i}$. Note that we have the optimal makespan $t^* \geq 1$. 
    
    We propose a schedule $\S$ of the bags $\B$ that has makespan at most $\max \{ 2, \beta(\B) \}$. The schedule $\S$ first schedules the large singleton bags in $\C$. For each bag in $\C$, the schedule $\mathcal{S}$ schedules the bag on the same machine as $\mathcal{S}^*$ places the corresponding jobs. We let $L_i$ be the total processing time of jobs scheduled on machine $i$ after all bags in $\C$ are scheduled. We then define the ``loading capacity" $C_i$ which is the upper bound on the total processing time of all jobs on machine $i$:
    \begin{equation*}
    C_i =\begin{cases}
         \max \{2,\beta (\B) \} s_i  \quad &\text{if} \, L_i \leq \max \{2,\beta (\B) \} s_i  \\
         \max \{2,\beta(\B)\} s_i + L_i  \quad &\text{if} \, L_i > \max \{2,\beta(\B) \} s_i \\
    \end{cases}
    \end{equation*}
    
    The schedule $\S$ sorts the bags in $\B \setminus \C$ in decreasing order of the bag processing times, and iteratively assigns each bag to the least loaded machine without violating the loading capacity bound of each machine.

    First, we argue that as long as every loading capacity $C_i$ is not violated for each machine $i$ after we schedule the bags in $\B \setminus \C$ then the robustness is bounded by $\max \{2,\beta (\B)\}$. We look at the completion time of each machine $i$ and compare it to $t^*$. We first consider any machine $i$ on which we let $C_i = \max \{2, \beta (\B)\} s_i$. If the total processing time of the jobs scheduled on this machine is bounded by $\max \{2, \beta(\B) \} s_i$ after we schedule all bags, then the ratio of the completion time on machine $i$ to $t^*$ is bounded by $\max \{2, \beta (\B)\}$ because we know $t^* \geq 1$. 
    Otherwise, if on machine $i$ we let $C_i = \max \{2, \beta(\B) \} s_i + L_i$ then this means that the total processing time of the large singleton bags in $\C$ scheduled on this machine $i$ is larger than $\max \{2, \beta(\B) \} s_i$. Because $\S$ has assigned the large singleton bags according to $\mathcal{S}^*$, we can obtain that $t^* \geq L_i / s_i$. If we satisfy the loading capacity bound after scheduling all bags in $\B$, then the completion time of this machine $i$ is bounded by $\frac{\max \{2, \beta(\B) \} s_i + L_i}{s_i} = \max \{2, \beta(\B) \} + \frac{L_i}{s_i}$ while $t^* \geq L_i / s_i$. Then the ratio of the completion time on machine $i$ to $t^*$ is bounded by $\frac{\max \{2, \beta(\B) \} + L_i/s_i}{L_i / s_i} = 1 + \max \{2, \beta(\B) \} s_i / L_i < 2$ since we have $L_i > \max \{2, \beta(\B) \} s_i$.
    
    Therefore, if we consider the ratio of each machine $i$'s completion time to $t^*$, the ratio is proven to be at most $\max \{ 2,\beta(\B) \}$. This suggests that the robustness of the partition $\B$ is at most $\max \{ 2, \beta(\B) \}$ as long as we satisfy the loading capacity bound $C_i$ defined above for each machine $i$.
    
    Next, we will argue that $\S$ can schedule all bags in $\mathcal{B}$ without violating the loading capacity bound. Clearly the loading capacity bound is not violated when $\S$ schedules the bags in $\C$. Now we show that when $\S$ schedules bags in $\B \setminus \C$, it's not violated either. 

    Assume by contradiction that there is a bag which cannot be assigned by $\S$ without violating a loading capacity constraint. Consider the first such bag and let $b$ be its total processing time. We also let $b'$ be the minimum total processing time of a bag in $\B$. Let $u$ be the number of bags that have been assigned already, with $u \geq |\C|$. We will give a lower bound on the total remaining loading capacity of the $m$ machines when the scheduling algorithm $\S$ fails to place the $(u + 1)$-th bag. Assume without loss of generality that $B_1, \ldots B_u$ are the first $u$ bags to be scheduled, and we fail to schedule $B_{u+1}, \ldots, B_m$, with $p(B_{u+1}) = b$. We let $V_p = \sum_{j=1}^{u}p(B_j)$ denote the total processing time of bags which have been scheduled and $V_{\ell} = \sum_{j=u+1}^{m}p(B_j)$ denote the total processing time of bags which are not scheduled. Note that $V_p \geq ub$ because $\S$ sorts the bags in $\B \setminus \C$ in decreasing order of the bag processing times, so all assigned bags have total processing time at least $b$.
    
    When $\beta(\B) \geq 2$, by assumption $\forall i > u$ we have $p(B_i) \leq \beta(\B) b'$. In particular we have $b \leq \beta(\B) b'$. Therefore the $(m - u)$ bags that were not placed have a total processing time $V_{\ell} \geq (m - u - 1)b' + b  \geq  (m-u-1+ \beta(\B) ) \frac{b}{\beta(\B)} \geq (m-u+1)\frac{b}{\beta(\B)}$. The remaining loading capacity when $\beta(\B) \geq 2$ is at least
    \begin{align*}
        \beta(\B) (\sum_{j \in [n]} p_j) - V_p  & =  \beta(\B) (V_{\ell} + V_p) - V_p \\
        & = \beta(\B) V_{\ell} + (\beta(\B) - 1 )V_p \\
        & \geq \beta(\B) (m - u + 1 )\frac{b}{\beta(\B)}+ (\beta(\B)-1)V_p \\
        & \geq (m - u + 1)b + (\beta(\B)-1)ub \\
        & > mb.
    \end{align*}

    When $\beta(\B) < 2$, then $\frac{b}{b'} \leq \beta(\B)$ implies that  $\frac{b}{b'} \leq 2$, and we have $V_{\ell} = (m - u - 1)b' + b  \geq  (m-u-1)\frac{b}{2} + b = (m-u+1) \frac{b}{2}$. Thus, the remaining loading capacity when $ \beta(\B) < 2$ is at least
    \begin{align*}
        2 (\sum_{j \in [n]} p_j) - V_p  & = 2(V_{\ell} + V_p) - V_p \\
        & = 2 V_{\ell} - V_p \\
        & \geq 2(m - u + 1 )\frac{b}{2}+ V_p \\
        & \geq (m - u + 1)b  + ub \\
        & > mb.
    \end{align*}

    Therefore, there must exist a machine with remaining loading capacity at least $b$ which contradicts the assumption that the bag of total processing time $b$ cannot be scheduled. Thus, we have shown that all bags in $\B$ can be scheduled by $\S$ without violating the loading capacity constraint $C_i$ of each machine $i$, which guarantees a robustness ratio $\max \{ 2, \beta(\B) \}$.
\end{proof}


\noindent  \textbf{\cref{lem-LPT}.} 
\textit{For any job processing times $p_1, \ldots, p_n$ and number of machines $m$, the partition $\B_{\textsc{LPT}} = \{B_1, \cdots, B_m\}$ returned by the LPT algorithm on these jobs  satisfies $\beta(\B_{\textsc{LPT}} ) \leq 2$.}
\begin{proof}
    Without loss of generality, let $B_1 = \argmax_{B \in \B_{\textsc{LPT}}, |B| \ge 2}p(B)$ and $B_2 = \argmin_{B \in \B_{\textsc{LPT}}} p(B)$. The LPT algorithm first sorts jobs in non-increasing order according to their processing times, and it places the job in the least loaded bag. Note that since $B_1$ has at least two jobs, the smallest job in $B_1$ has processing time at most $\frac{p(B_1)}{2}$. Then by the rule of the LPT algorithm, we have $\frac{p(B_1)}{2} \leq p(B_2)$; otherwise, $p(B_2) < \frac{p(B_1)}{2}$ so when we schedule the smallest job in $B_1$, the bag $B_{2}$ is less loaded than $B_1$. There is a contradiction. 
\end{proof}


\noindent  \textbf{\cref{lem-simple}.} \textit{For any job processing times $p_1, \ldots, p_n$ and partition $\B = \{B_1, \cdots, B_m\}$ of the jobs, if  $\beta(\B)  \leq 2$, then $\min_{B \in \B} p(B) \geq \frac{\sum_{j=1}^n p_j}{2m-1}.$}
\begin{proof}
    The proof is by contradiction. Assume that $b' = \min_{B \in \B}p(B) < \frac{\sum_{j=1}^n p_j}{2m-1}$, and assume without loss of generality that $p(B_m) = b'$. Then $p(B_i) \le 2b' < 2\frac{\sum_{j=1}^n p_j}{2m-1}$ for $i \in [m-1]$. 
    The total processing time of bags in $\B$ is $\sum _{i \in [m]} p(B_i) = p(B_m) + \sum_{i \in [m-1]} p(B_i) < \frac{\sum_{j=1}^n p_j}{2m-1} + (m-1)\cdot 2 \frac{\sum_{j=1}^n p_j}{2m-1} =  \sum_{j=1}^n p_j$. Contradiction. Therefore, $b' \geq \frac{\sum_{j=1}^q p_j}{2m-1}$.
\end{proof}


\noindent  \textbf{\cref{thm-general}.} \textit{Consider the algorithm that runs \ipb \ with $\rho = 4$ in the partitioning stage and a PTAS for makespan minimization in the scheduling stage.  For any constant $\epsilon \in (0,1)$ and any $\alpha \in (0,1)$, this algorithm achieves a 
$\min\{\eta^2(1+\epsilon)(1+\alpha), (1+\epsilon)(2 + 2/\alpha)\}$ approximation for SSP where $\eta = \max_{i \in [m]} \frac{\max\{\hat{s}_i, s_i\}}{\min\{\hat{s}_i, s_i\}}$ is the prediction error.}

\begin{proof}
Recall that $\s$ is the true speed configuration and $\spred$ is the prediction. Let $opt(\p,\s)$ and $opt(\p,\spred)$ be the optimal makespan when the machine speeds are $\s$ and $\spred$ to schedule the individual jobs with processing times $\p$, respectively. Assume that if we run \ipb \ with $\rho = 4$ in the partitioning stage with the predicted speeds $\spred$, the returned partition is $\bagsipr = \{B_1,\ldots,B_m\}$. The algorithm then receives the true speeds $\s$ as the input to a PTAS to schedule $\bagsipr$.

Let $\epsilonp \in (0,1)$ be a constant such that $(1+\epsilonp)^2 = 1 + \epsilon$. By Lemma~\ref{lem:consistency}, we know $\bagsipr$ is a $(1+\alpha)(1+\epsilonp)$-consistent partition obtained from \ipb \ with $\rho = 4$ and accuracy $\epsilonp$. This means that there exists a schedule $\S_{cons}$ of the partition $\bagsipr$ on the hypothetical machines of speeds $\hat{s}_1, \ldots, \hat{s}_m$ such that the makespan of scheduling $\bagsipr$ with $\S_{cons}$ on the hypothetical machines is bounded by $(1+\alpha)(1+\epsilonp)opt(\p,\spred)$. We use $t(\S_{cons}, \spred)$ to denote the makespan of $\bagsipr$ scheduled with $\S_{cons}$ on the hypothetical machines, and we have 
\begin{equation}\label{ineq-err-1}
    t(\S_{cons}, \spred) \leq (1+\alpha)(1+\epsilonp)opt(\p,\spred).
\end{equation}

Recall our algorithm actually uses PTAS to schedule $\bagsipr$ with the input speeds $\s$ in the scheduling stage. We let $\epsilon'$ be the accuracy of the PTAS for scheduling.
Let this schedule of $\bagsipr$ be $\mathcal{S}_{PTAS}$. Now, consider the schedule $\mathcal{S}_{cons}$ when the true speed is $\s$. If we use $\mathcal{S}_{cons}$ instead of $\mathcal{S}_{PTAS}$ to schedule the partition $\bagsipr$, then the makespan $t(\mathcal{S}_{cons}, \s)$ of the partition $\bagsipr$ under machine speeds $\s$ satisfies
\begin{equation}\label{ineq-err-2}
    t(\mathcal{S}_{cons}, \s) \leq \eta t(\mathcal{S}_{cons}, \spred)
\end{equation}
because $\s$ and $\spred$ are off by a factor at most $\eta$.

Lastly, consider the makespan of $\bagsipr$ with schedule $\mathcal{S}_{PTAS}$ when the true speed is $\s$. This is the makespan achieved by our algorithm described in the theorem when the true speed is $\s$. Let this makespan be denoted by $t(\mathcal{S}_{PTAS}, \s)$. Then we have 
\begin{equation}\label{ineq-err-3}
    t(\mathcal{S}_{PTAS}, \s) \leq (1+\epsilonp)t(\mathcal{S}_{cons}, \s)
\end{equation}
because PTAS is an $(1+\epsilonp)$ approximation for scheduling the partition $\bagsipr$.

Combining the inequalities (\ref{ineq-err-2}) and (\ref{ineq-err-3}), we have 
\begin{align}
    t(\mathcal{S}_{PTAS}, \s) &\leq  (1+\epsilonp)t(\mathcal{S}_{cons}, \s)  \leq \eta(1+\epsilonp)t(\mathcal{S}_{cons}, \spred). \label{ineq-err-combined}
\end{align}

Next, we consider the relationship between $opt(\p,\s)$ and $opt(\p,\spred)$. We prove $opt({\p,\s}) \geq  \frac{1}{\eta} opt(\p,\spred)$ by contradiction. Assume for the sake of contradiction that $opt(\p,\s) < \frac{1}{\eta} opt(\p,\spred)$. When the true speed configuration is $\s$, let the optimal algorithm have partition $\B_s^{*}$ and schedule $\mathcal{S}_s^*$. When the true speed configuration is $\spred$, let the optimal algorithm have partition $\B_{\hat{s}}^{*}$ and schedule $\mathcal{S}_{\hat{s}}^*$. Then, when the true speed configuration is $\spred$, if we simply use the same partition $\B_s^{*}$ and schedule $\mathcal{S}_s^*$, we can obtain a makespan at most $\eta \cdot opt(\p,\s) < opt(\p,\spred)$ by our assumption. This contradicts the optimality of $(\B_{\hat{s}}^{\ast}, \mathcal{S}_{\hat{s}}^*)$. 
Thus,
\begin{align}\label{ineq-err-opt}
    opt(\p,\s) \geq \frac{1}{\eta} opt(\p,\spred).
\end{align}

Therefore, we have 
\begin{align*}
    \frac{t(\mathcal{S}_{PTAS}, \s)}{opt(\p,\s)} & =  \frac{t(\mathcal{S}_{PTAS}, \s)/t(\mathcal{S}_{cons},\spred)}{opt(\p,\s)/opt(\p,\spred)}\frac{t(\mathcal{S}_{cons},\spred)}{opt(\p,\spred)} & \\
    & \leq \frac{(1+\epsilonp)\eta}{1/\eta} \cdot \frac{t(\mathcal{S}_{cons},\spred)}{opt(\p,\spred)} & \text{by inequalities~(\ref{ineq-err-combined}) and (\ref{ineq-err-opt})}\\
    &  \leq \eta^2 (1+\alpha) (1+\epsilonp)^2. & \text{by inequality~(\ref{ineq-err-1})}
\end{align*}

The above proves that our algorithm is an $\eta^2 (1+\alpha) (1+\epsilonp)^2 = \eta^2 (1+\alpha) (1+\epsilon)$ approximation because we initially set $(1+\epsilonp)^2 = 1+\epsilon$. 

Finally, from \cref{lem:robustness}, we know that our algorithm with accuracy $\epsilonp$ is a $(2 + 2/\alpha)(1+\epsilonp)$ approximation because the partition returned by \ipb \ with $\rho =4$ is $(2+2/\alpha)$-robust and we pay an extra $(1+\epsilonp)$ factor in the scheduling stage. Since $\epsilonp < \epsilon$, the approximation ratio is clearly also bounded by $(2 + 2/\alpha)(1+\epsilon)$.

Thus, we conclude that the algorithm is a $\min \{ \eta^2  (1+\alpha) (1+\epsilon), (2 + 2/\alpha)(1+\epsilon)\} $ approximation. 
\end{proof}



\noindent \textbf{\cref{runtime}.} \textit{At most $O(m^2)$ iterations are needed for \ipb \ with $\rho = 4$ to terminate.}

\begin{proof}
Recall that in the $\lptrb$ subroutine, the collection of bags $\Mmax$ receives $\Bmin = \arg \min _{B \in \cup_{i \in [m]} \M_i}p(B)$ and gets balanced. For notational convenience, we let $\bmax = \max _{B \in \Mmax, |B| \geq 2} p(B) $ be the maximum total processing time of a non-singleton bag in $\Mmax$, and let $\bmin = p(\Bmin)$. Note that $\Mmin$ is the collection of bags that contains the bag $\Bmin$, which gets transferred to $\Mmax$.

We first look at a particular collection of bags $\M_j$ that is maintained by the \ipb \ algorithm. Recall that for any quantity $q$ of our interest, we use $q^{(i_k)}$ to denote the quantity $q$ in the $i_k^{th}$ iteration. We claim that if a collection of bags $\M_j^{(i_1)} = \Mmin$ in the $i_1^{th}$ iteration, then if $\M_j^{(i_2)} = \Mmax$ for the first time in the $i_2^{th}$ iteration with $i_2 > i_1$, we have $\bmax^{(i_2)} / \bmin^{(i_2)} \leq 2$. To start with, we consider the $i_1^{th}$ iteration, where $\Bmin^{(i_1)} \in \M_j^{(i_1)}$. Let $B^* = \arg \max _{B \in \M_{j}^{(i_1)}, |B| \geq 2} \{ p(B) \}$ be a non-singleton bag of the largest processing time in $\M_{j}$, and let $b^* = p(B^*)$. Note that we have omitted the $(i_1)$ notation from $B^*$ and $b^*$, because from the $i_1^{th}$ iteration to the $i_2^{th}$ iteration, the longest non-singleton bag's processing time in $\M_j$ is not changed. 
Recall that the algorithm only adds bags to and balances the collection of bags $\Mmax$. As $i_2$ is the first time this collection of bags $\M_j$ becomes $\Mmax$ after the $i_1^{th}$ iteration, it cannot receive any bag from the iterations in between $i_1$ and $i_2$. There's no way to create another non-singleton bag of longer processing time in $\M_j$ during the iterations between $i_1$ and $i_2$ as it does not receive anything. However, it is possible that we removed some bags from $\M_j$ during the iterations; but since the algorithm always removes a bag $\Bmin$ of the minimum processing time, having removed all bags of total processing time $b^*$ means that all the non-singleton bags in $\M_j$ have been removed. If there does not exist any non-singleton bag in the assignment, it won't be chosen to be $\Mmax$ by the construction of the algorithm. Thus, we conclude that there must exist at least one bag of size $b^*$ that stays in $\M_j$, and this bag is the non-singleton bag of the longest processing time across iterations $i_1, \ldots, i_2$. 

We now look inside $\M_j$ in the $i_1^{th}$ iteration. By lemma~\ref{lem-LPT}, we have $b^* \leq 2\bmin^{(i_1)}$, because each collection of bags either only contains one bag or has been balanced by $\lptrb$. Then, in the $i_2^{th}$ iteration, as $\Mmax^{(i_2)} = \M_j^{(i_2)}$ by our assumption, we have $\bmax^{(i_2)} = b^* \leq 2\bmin^{(i_1)} \leq 2\bmin^{(i_2)}$ by lemma~\ref{lem:non-decre}. Therefore, in the $i_2^{th}$ iteration, we have $\bmax^{(i_2)} \leq 2\bmin^{(i_2)}$ and the algorithm will be terminated without actually balancing the collection of bags $\M_j$ in the $i_2^{th}$ iteration. 

We have proved in the above that if a collection of bags in an assignment becomes $\Mmin$ at some point, then later if this collection of bags becomes $\Mmax$, the algorithm will terminate. In other words, once a collection becomes $\Mmin$, the algorithm won't be able to add any bag to the collection without terminating itself. Because there are only $m$ bags, each collection of bags in an assignment can be $\Mmin$ for at most $m$ times before it becomes empty (and it stays empty ever after) or becomes $\Mmax$ to receive some bag, which leads to the termination of the algorithm. Thus, a collection can only be $\Mmin$ for at most $m$ times throughout the algorithm. Finally, since there are $m$ collections of bags in total, and each collection of bags can be $\Mmin$ for at most $m$ times, we have at most $m^2$ iterations.
\end{proof}

\section{Additional Details of the Experiment Setups}
\label{sec:appexperiments}



\paragraph{Data generation.} In Section~\ref{sec:experiments}, we have specified how we sample the job processing times $p_j$, the machine speeds $s_i$ and the prediction errors $err(i) = \hat{s}_i - s_i$, which are drawn i.i.d. either from the Normal distribution or the Uniform distribution. There are some occurrences, although an extremely small number of them, where  the processing times or speeds sampled are negative. If $p_j$, $s_i$, or $\hat{s}_i$ is negative in the sample, we replace its value by $10^{-3}$. 

\paragraph{Computing a solution to the makespan minimization problem.} We have mentioned in Section~\ref{sec:experiments} that we give advantages to the benchmarks by using an IP to compute the schedule in the scheduling stage. Any schedule of bags (or jobs) on the machines can be expressed as a matrix of indicator variables $x_{ji}$, where $x_{ji} = 1$ if bag (or job) $j$ is assigned to machine $i$ and 0 otherwise. A feasible schedule must satisfy $\sum _{i=1}^m x_{ji} = 1, \ \forall j$. The IP minimizes the makespan over all feasible schedules and returns the optimal schedule within the integrality gap, which is set to be  $1\%$. 

The time limit of the IP solver (Gurobi) is set to be 120 seconds. The time-out frequency of Gurobi is extremely low, and the maximum final integrality gap is never more than $4\%$; therefore, we adopted the solution returned by the IP even if the solver timed out. Each of the plots in \cref{fig:exp} requires the construction of $1100$ instances, on which we tested our algorithm. In each of the experiments we ran, excluding the one that varies the number of machines $m$, we experienced time-out for no more than $10$ times on $1100$ instances. 


To ensure that our algorithm has a polynomial running time,  \ipb \ uses the LPT algorithm to compute a schedule instead of the IP. Given a set of bags (or jobs) and $m$ machine speeds as inputs, the LPT scheduling algorithm first sorts the total processing times of the bags (or jobs) in decreasing order and then assigns them to the machines iteratively. In each iteration, the LPT scheduling algorithm iterates over all machines and computes the makespan if the current bag (or job) $j$ is assigned to machine $i = 1,2, \ldots, m$. Then the LPT scheduling algorithm assigns the bag (or job) $j$ to the machine $i$ that minimizes the makespan. Once $j$ is assigned to machine $i$, it will be on that machine permanently and the algorithm moves on to the next iteration. 
\end{document}